\documentclass[11pt]{article}

\usepackage[left=2cm,right=2cm,top=1in, bottom=1in]{geometry}
\usepackage{amsmath, amsthm, amsfonts, enumitem, amssymb, url, xcolor, booktabs}
\usepackage{adjustbox}
\usepackage{xspace}
\usepackage{setspace}
\usepackage{graphicx}
\usepackage{tabularx}
\usepackage[utf8]{inputenc}

\makeatletter
\newenvironment{innerproof}[1][\proofname]
  {\par\normalfont \topsep6\p@ \@plus6\p@\relax
  \trivlist
  \item[\hskip\labelsep\itshape#1\@addpunct{.}]\ignorespaces}
  {\endtrivlist\@endpefalse}
\makeatother

\newtheorem{observation}{Observation}
\newtheorem{proposition}{Proposition}
\newtheorem{lemma}{Lemma}
\newtheorem{claim}{Claim}
\newtheorem{corollary}{Corollary}
\newtheorem{definition}{Definition}
\newtheorem{theorem}{Theorem}
\newcommand{\nmgl}{$\ell$-NCG}
\newcommand{\nmglp}{$\ell'$-NCG}

\title{Undetectable Selfish Mining}
\author{Maryam Bahrani\thanks{a16z crypto
(\texttt{maryam.bahrani.14@gmail.com}).} \and S. Matthew Weinberg\thanks{Princeton University (\texttt{smweinberg@princeton.edu}).}}
\date{}

\begin{document}

\maketitle
\begin{abstract}
   
Seminal work of Eyal and Sirer~\cite{EyalS14} establishes that a strategic Bitcoin miner may strictly profit by deviating from the intended Bitcoin protocol, using a strategy now termed \emph{selfish mining}. More specifically, any miner with $>1/3$ of the total hashrate can earn bitcoin at a faster rate by selfish mining than by following the intended protocol (depending on network conditions, a lower fraction of hashrate may also suffice). 

One convincing critique of selfish mining in practice is that the presence of a selfish miner is \emph{statistically detectable}: the pattern of orphaned blocks created by the presence of a selfish miner cannot be explained by natural network delays. Therefore, if an attacker chooses to selfish mine, users can detect this, and this may (significantly) negatively impact the value of BTC. So while the attacker may get slightly more bitcoin by selfish mining, these bitcoin may be worth significantly less USD.

We develop a selfish mining variant that is provably \emph{statistically undetectable}: the pattern of orphaned blocks is statistically identical to a world with only honest miners but higher network delay. Specifically, we consider a stylized model where honest miners with network delay produce orphaned blocks at each height independently with probability $\beta'$. We propose a selfish mining strategy that instead produces orphaned blocks at each height independently with probability $\beta > \beta'$. We further show that our strategy is strictly profitable for attackers with $38.2\% \ll 50\%$ of the total hashrate (and this holds for all natural orphan rates $\beta'$). 
\end{abstract}
\newpage

\section{Introduction}\label{sec:intro}
At their core, cryptocurrencies and other blockchain protocols aim to provide a totally-ordered \emph{ledger} of transactions that can serve as a record of payments. Key to their value proposition is that these ledgers are maintained in a \emph{decentralized manner} --- anyone can join their network and start authorizing transactions (in Bitcoin, this process is referred to as \emph{mining} and the participants as \emph{miners}). However, miners must be incentivized in order to authorize transactions, and in particular to follow the intended protocol.

Nakamoto's 2008 whitepaper~\cite{Nakamoto08} introducing Bitcoin applies a clever technique termed \emph{proof-of-work} --- miners are selected to authorize the next block of transactions proportionally to their computational power (often referred to as their \emph{hashrate}, because they use this computational power to repeatedly compute the hash function SHA-256). That same whitepaper already observes that incentivizing miners to follow the intended protocol can be tricky, however. For example, Nakamoto observes that a miner with $>1/2$ of the total hashrate in the network can strictly profit by deviating from the Bitcoin protocol. However, for several years it was largely assumed that $50\%$ was the cutoff: as long no miner controlled more than half the total hashrate in the network, everyone should follow the intended protocol.

Eyal and Sirer prove this assumption false with their seminal paper introducing the Selfish Mining strategy (we will overview their strategy in detail in Section~\ref{sec:prelim})~\cite{EyalS14}. Their strategy is strictly profitable compared to honestly following the Bitcoin protocol even with just $>1/3$ of the total hashrate. Followup work of Sapirshtein, Sompolinsky, and Zohar later provides an improved strategy that is strictly profitable with just $\approx 32.9\%$ of the total hashrate, which is tight (for a strategic miner with poor network connectivity)~\cite{SapirshteinSZ16}. With better network connectivity, Selfish Mining can be profitable with arbitrarily small fraction of hashrate --- see discussion in Section~\ref{sec:prelim}.

From a mechanism design perspective, these facts initially appear somewhat alarming: Depending on the hashrate of the largest miner, and also on network conditions, following the protocol is not a Nash equilibrium. Moreover, the Selfish Mining attack is in a sense untraceable: because Bitcoin is pseudonymous, there is no way to distinguish among identities of block creators to determine which blocks were created by a dishonest miner. That is, depending on the hashrate of the largest miner and network conditions, there is an untraceable and strictly profitable deviation from the Bitcoin protocol.

However, one key critique of Selfish Mining (and other strategic deviations) in practice is its \emph{statistical detectability}. That is, the pattern of blocks created in a world with a Selfish Miner does not have any ``innocent explanation'' due to network latency, and it will be clear that \emph{someone} is Selfish Mining.\footnote{To clarify: it may not be clear who is Selfish Mining, nor how to punish them, but it will be clear that \emph{someone} is Selfish Mining. Note also that collecting the relevant data needed to run the statistical test is non-trivial.} Should Selfish Mining be detected, it is possible that the value of the underlying cryptocurrency will suffer significantly. Therefore, while Selfish Mining may be profitable when denominated in the underlying cryptocurrency (e.g. BTC), it could be wildly unprofitable when denominated in an objective unit of value (e.g. USD).\footnote{Note that even if a deviator were to hedge against this possibility by shortselling the underlying cryptocurrency, there is still significant uncertainty in how the market might react (which may be sufficient to dissuade a risk-averse player, even if hedges exist to make deviating profitable in expectation).}

The \textbf{main result of our paper introduces a \emph{statistically undetectable Selfish Mining} strategy} in Theorems~\ref{thm:main1} and~\ref{thm:main2}. That is, the pattern of block creation induced by our strategy is statistically identical to that of a world with only honest miners but increased latency. The mathematics behind our results apply to the canonical setting studied in Eyal and Sirer's seminal work: proof-of-work longest-chain protocols with block rewards.\footnote{Our techniques may certainly be relevant more broadly to longest-chain protocols (e.g. with transaction fees instead of block rewards, or proof-of-stake instead of proof-of-work), simply because these settings have significant technical overlap. Still, we do not claim that any of our formal results carry to another setting, and applying our techniques in this setting may still require novel insight.} Still, we hope the reader takes away the following conceptual point much more broadly: attackers concerned about the impact of their attack on a cryptocurrency's value may seek an undetectable variant. Therefore, designers must consider this possibility when claiming incentive compatibility of protocols (and especially if those claims reference a strategic deviation's impact on the cryptocurrency's value). 

We discuss the implications of our results in more detail in Section~\ref{sec:conclusion}. After overviewing related work below in Section~\ref{sec:related}, we proceed with a detailed description of Selfish Mining, along with a description of our model in Section~\ref{sec:prelim}. We follow this with formal statements of our main results in Section~\ref{sec:results}. Section~\ref{sec:warmup} proves a warmup result to demonstrate some key ideas, and Section~\ref{sec:main1} proves our first main result. Appendix~\ref{app:whynot} contains a summary of discussions regarding Selfish Mining in practice. Extensions of our main results and omitted proofs are contained in later appendices.

\subsection{Related Work}\label{sec:related}

\paragraph{Strategic Manipulations of Consensus Protocols.} Following Eyal and Sirer's seminal work~\cite{EyalS14}, there have been numerous papers on Selfish Mining. Most relevant to our work are those that understand the limits of its profitability, such as~\cite{KiayiasKKT16, SapirshteinSZ16}. In analyzing our most general result, we use tools initially developed in~\cite{SapirshteinSZ16}.

Strategic manipulation of consensus protocols via some form of ``block withholding'' is studied in a wide array of related domains. Those most similar to our work consider deviations in proof-of-work longest-chain protocols with transaction fees~\cite{CarlstenKWN16}, proof-of-stake longest-chain protocols~\cite{BrownCohenNPW19, NeuderMRP19, NeuderMRP20, FerreiraW21}, and leader-selection protocols~\cite{FerreiraHWY22}. Our work bears some similarity to these in that we also study strategic manipulation of consensus protocols via block withholding, but the models are distinct.

Other works examine different styles of profitable deviations from Proof-of-Work blockchain protocols. For example,~\cite{FiatKKP19, GorenS19,YaishTZ22,YaishSZ22} consider manipulating difficulty adjustments to earn extra rewards.~\cite{YaishTZ22,YaishSZ22} do so in particular by manipulating the timestamp of created blocks.

\paragraph{Strategic Manipulation in Practice.} Among the broad line of works referenced above,~\cite{YaishSZ22} is especially impactful, as they also provide clear evidence of strategic manipulation in practice (which is the first such evidence, to the best of those authors' and these authors' knowledge). It is important for future work to understand the impact that their discovery had on  the deviator (F2Pool) and/or the underlying currency (ETH) -- this would quantitatively impact any risk assessment by miners considering strategic manipulation. In relation to our work, this would impact the importance of statistical undetectability to strategic miners.

We've informally used the term ``strategic manipulation'' to describe deviations from an intended consensus protocol that net the deviator additional rewards \emph{issued directly by the consensus protocol} in the form of block rewards. There is also a significant history of block producers acting selfishly to profit from economic activity on top of consensus protocols (e.g. double-spend attacks.)

\paragraph{Incentives beyond the Consensus Layer.} Beyond incentive compatibility concerns arising from block rewards, there is also a growing literature considering incentive design in other aspects of blockchain protocols. One example is recent work on transaction fee mechanism design (auction design for the inclusion of transactions within a block in the presence of strategic block producers)~\cite{LaviSZ19, Yao18, HubermanLM21,Roughgarden21,bahrani2023-TFM-MEV,ChungS23,ShiCW23}. These works also fit into a growing literature on mechanism design within consensus protocols, but otherwise have no technical overlap with ours.

\paragraph{Detection of Selfish Mining in practice.} We are aware of two lines of work related to statistical detection of Selfish Mining. One proposes statistical tests (such as the length of orphaned chains -- longer orphan chains are claimed to be indicative of Selfish Mining), and perhaps evaluates them in a simulation environment~\cite{ChicarinoAJR20, WangLLWY21, SaadNKM19}. To the best of our knowledge, most of these tests have not been deployed in practice, and therefore do not offer an opinion on whether or not Selfish Mining is occurring in practice.

A smaller line of work aims to detect Selfish Mining empirically.~\cite{LiCT22} claims to be the first paper to detect selfish mining in practice, building upon previous work of similar authors~\cite{LiYT20,LiYT20b}. Their statistical test looks exclusively at the main chain, and focuses on the distribution of consecutive blocks mined by the same miner. Their null hypothesis is that the miner of each block on the main chain should be independent of prior blocks, whereas a Selfish Miner would disproportionately mine their blocks consecutively. They claim that observed mining patterns fail their null hypothesis to varying degrees across five blockchains: Bitcoin, Litecoin, Ethereum PoW, Monacoin, and Bitcoin Cash. But, the authors acknowledge that their null hypothesis could fail for other reasons. For example, note that their null hypothesis would also fail when orphans naturally occur and miners tiebreak in favor of their own blocks (this is because a block competing in a natural fork is more likely to wind up in the longest chain when its creator finds the next block).

\section{Background}\label{sec:prelim}
\subsection{Nakamoto Consensus and the Longest Chain Protocol}

Bitcoin and many other altcoins use a \emph{longest-chain protocol} with \emph{proof-of-work}, which is often referred to as \emph{Nakamoto Consensus}. We overview relevant details of the protocol below, and refer the reader to resources such as~\cite{narayanan2016bitcoin} for an explanation of why these details capture the Bitcoin protocol. Note that the features we highlight below are exactly the same features from~\cite{EyalS14} --- the stylized model we pose below is not new, and is extensively studied following~\cite{EyalS14}.

\paragraph{Nakamoto Consensus Game (NCG)} There is a set $N$ of $n$ miners, and miner $i$ has hashrate $\alpha_i > 0$. Time proceeds in discrete steps $t = 1,2,\ldots$. At the start of each timestep $t$, there is a set of blocks $B_t$ that have been previously broadcast, and a set of blocks $B_t^i$ that have been previously created by miner $i$ (perhaps broadcast, perhaps not). Initially, $B_1^i$ is empty for all $i$, and $B_1$ contains a single ``genesis block.'' 
At every time step $t$:
\begin{itemize}
		\item A miner $m_t$ is selected so that $m_t :=i$ with probability $\alpha_i/(\sum_{j=1}^n \alpha_j)$, independently of all previous rounds. $m_t$ creates a \emph{block} $b_t$, and $m_t$ chooses the block to which $b_t$ \emph{points}: $m_t$ may choose any block in $B_t \cup B_t^{m_t}$ (that is, the newly created block must point to exactly one block, and that block can be any block that was broadcast before time $t$, or created by $m_t$ before time $t$, or both).
		\item Every miner $i$ can choose to broadcast any blocks in $B_{t+1}^{i}\setminus B_t$ (that is, every miner $i$ can broadcast any block that they created in a timestep $\leq t$ and have not already broadcast).
\end{itemize}
We refer to the \emph{height} of a block $h(b_t)$ as the number of blocks in the path leaving $b_t$ (i.e. by following pointers). A \emph{Longest Chain} at time $t$ is any block in $B_t$ of greatest height. If the longest chain at time $t$ is unique, we denote by $R_t^i$ to be the fraction of blocks in the longest chain that were created by miner $i$. If the longest chain at time $t$ is not unique, let $t'< t$ denote the most recent timestep where the longest chain at $t'$ is unique, and define $R_t^i:=R_{t'}^i$. Observe that the longest chain at time $1$ is unique, so this is always well-defined. Player $i$'s reward is $\lim\inf_{t \rightarrow \infty} R_t^i$.\footnote{Our work and all prior works only consider strategy profiles where the limit exists, so the distinction between $\lim$ and $\liminf$ is just a formality.}

In particular, observe in this game that every miner has two decisions every round: (a) if they create a block, what does it point to? and (b) what previously-created blocks do they broadcast? Observe also that miners are paid proportionally to the steady-state fraction of blocks they produce in the longest-chain --- this captures the fact that miners are paid primarily according to a block reward that comes with each block, and that Nakamoto Consensus has a difficulty adjustment that causes the length of the longest chain to grow proportionally to time (we again refer the reader to~\cite{narayanan2016bitcoin, EyalS14} for further details on connecting this game to Bitcoin and related cryptocurrencies).

\paragraph{Longest Chain Protocol} The longest-chain protocol asks each miner to use a specific strategy in this game: Whenever you create a block, point to a longest chain, and immediately broadcast your block. We call a miner \emph{honest} if they follow the longest-chain protocol, and \emph{strategic} otherwise. The key question is whether it is a Nash equilibrium for every miner to follow the longest-chain protocol.

It is folklore knowledge from Nakamoto's whitepaper~\cite{Nakamoto08} that the longest-chain protocol is not a Nash equilibrium if $\alpha_i > \sum_j \alpha_j/2$ for any $i$. In that case, Player $i$ can simply ignore all blocks created by other miners and point only to their own most recent block (broadcasting all blocks upon creation). Then Player $i$ will have reward $1$, and all other players have reward $0$.\footnote{Note that, beyond the fact that Player $i$ is accumulating all of the mining rewards, they control the entire content of the ledger, and can launch significantly more malicious attacks.} However, until Eyal and Sirer's seminal work, it was assumed that the longest-chain protocol is indeed a Nash equilibrium when $\alpha_i < \sum_j \alpha_j/2$ for all $i$.

\subsection{Selfish Mining}
We now overview the Selfish Mining strategy, and review its performance when all other miners are using the longest-chain protocol. Let's first consider the case where all other miners break ties in favor of the strategic miner $i$ (that is, if the longest chain is not unique, they will point their block towards the one that is created by miner $i$). Consider the following strategy in NCG:

\paragraph{Strong Selfish Mining:}
\begin{itemize}
\item When selected to mine a block, point to a longest-chain, breaking ties in favor of a block created by yourself.
\item \emph{Do not immediately broadcast this block.} Instead, during any round $t$ where you are not selected to mine, and the longest chain before round $t$ has height $h$, broadcast a block of height $h+1$ (if you have one).
\end{itemize}

\noindent If a strategic miner $i$ uses Strong Selfish Mining and all other miners follow the longest chain protocol, then the following occurs. First, \emph{every} block created by miner $i$ is broadcast during the same round as a block created by another miner of the same height. In other words, every height either contains a single block created by a miner $\neq i$, or two conflicting blocks broadcast during the same round, with one created by miner $i$. Second, because every miner tiebreaks in favor of miner $i$, all of $i$'s blocks will enter the longest chain (and conflicting blocks will not). Therefore:

\begin{theorem}[\cite{EyalS14}] If all miners $\neq i$ use a longest-chain protocol that tiebreaks in for Miner $i$, and Miner $i$ uses Strong Selfish Mining, then Miner $i$ gets reward $\frac{\alpha_i/\sum_j \alpha_j}{1-\alpha_i/\sum_j \alpha_j} > \alpha_i/\sum_j \alpha_j$.
\end{theorem}

Strong Selfish Mining is strictly profitable for any $\alpha_i > 0$, but relies on the fact that all other miners tiebreak in its favor. Eyal and Sirer connect tiebreaking in the above-described Nakamoto Consensus Game to network connectivity in practice. Really, a ``timestep'' in the Nakamoto consensus game corresponds to ``some miner has just succeeded in inverting SHA-256 and created a new block.'' An extremely well-connected miner can (a) listen for this block to be broadcast, and then (b) broadcast their own previously-withheld block, while still (c) ensuring that their own block arrives at other miners first. Mapping this onto the stylized Nakamoto consensus game yields Strong Selfish Mining. For the rest of this paper, we'll replace ``during any round $t$ where you are not selected to mine, and the longest chain before round $t$ has height $h$'' with ``during any round $t$ where another miner broadcasts a block of height $h+1$'' to emphasize this connection (and so that our language sounds less tailored to this specific stylization).

Of course, it is not realistic to assume that any strategic miner has such strong network connectivity. Eyal and Sirer further pose the (canonical) Selfish Mining strategy that is profitable even when the strategic miner \emph{loses all tiebreaks} --- this corresponds to a reality where the poorly-connected miner can still execute (a) and (b), but not (c). Instead, their canonical strategy (paraphrased below) relies only on the fact that honest miners will still always select a strictly longer chain over a strictly shorter one. 

\paragraph{Selfish Mining:}
\begin{itemize}
	\item When selected to mine a block, point to a longest-chain, breaking ties in favor of a block created by yourself.
	\item \emph{Do not (necessarily) immediately broadcast this block.} Instead, for each block $b$ of height $h$ that you create:
		\begin{itemize}
			\item At the moment where another miner broadcasts a block of height $h$, broadcast $b$.
			\item Or, during the moment/round where $b$ becomes \emph{pivotal}, broadcast $b$. A block of height $h$ is pivotal if (a) a block of height $h-1$ has been broadcast by another miner, (b) you created a block of height $h-1$, and (c) you have not yet created a block of height $h+1$.
		\end{itemize}
\end{itemize}


\noindent If a strategic miner $i$ uses Selfish Mining, and all other miners follow the longest-chain protocol, then the following occurs. At all points in time, there may be a unique longest chain with height $h$. If so, and a miner $\neq i$ creates the next block, it will be broadcast and the height increases to $h+1$. If instead $i$ creates the next block, $i$ will withhold the block and eventually create a conflict at height $h+1$. If there are multiple longest chains with height $h$, then miner $i$ may have created blocks that haven't been broadcast. If miner $i$ has only one such block, then it is pivotal, and miner $i$ will immediately broadcast it (to ensure that it wins the conflict, because it can only win with a strictly longer chain). If miner $i$ has multiple such blocks, it can safely wait until broadcasting, as long as they broadcast their last hidden block the moment it becomes pivotal. Intuitively, the Selfish Miner faces a risk compared to the longest-chain protocol: when withholding their first block, there is a chance that they do not find either of the next two blocks. In this case, their block is orphaned.\footnote{A block is ``orphaned'' if it is not in the eventual longest chain.} If they find exactly one of the next two blocks, then this block becomes pivotal and broadcast --- this allows the Strategic Miner to orphan another miner's block (thus increasing its own ratio). If instead they find both of the next two blocks, then they now have a lead of three withheld blocks, and will be able to orphan even more blocks of other miners (thus further increasing its own ratio). Because there is a risk, the strategy is only profitable for sufficiently large $\alpha_i$. 

\begin{theorem}[\cite{EyalS14}] If all miners $\neq i$ use a longest-chain protocol that tiebreaks \emph{against} Miner $i$, and Miner $i$ uses Selfish Mining, then Miner $i$ gets reward $>\alpha_i$ if and only if $\alpha_i > 1/3$.
\end{theorem}

\subsection{Detecting Selfish Mining}


Because Bitcoin is pseudonymous, an attacker can create a new public key for every block they create, so there is no record of the same miner creating multiple blocks.  With no identity attached to any particular block, there is little to distinguish the Selfish Miner's blocks from others.\footnote{But they are not \emph{fully} indistinguishable. The Selfish Miner's blocks will always be created before the competing blocks, and this can manifest in a few ways (including timestamps and transaction content). One might therefore propose to mitigate the profitability of Selfish Mining by asking honest miners to tiebreak in favor of blocks that were created most recently. Unfortunately, this creates significant incentive issues (that designers are well aware of): now an attacker need not build upon the current longest-chain because they can just replace it instead. In general, targeted punishments against blocks that could have been created by Selfish Mining may be worthwhile to explore, but this is very far from current norms (and may be impossible to do without creating new incentive issues).}


While a Selfish Miner may not be identifiable, their presence is statistically detectable. To help make this point more rigorous, consider the following generalized Nakamoto consensus game, parameterized by a latency $\ell$.\footnote{Appendix~\ref{app:latency} contains a brief discussion relating this stylized latency model to richer latency models considered in Distributed Systems.}

\paragraph{Nakamoto Consensus Game, parameterized by $\ell$ (\nmgl)}

\noindent The Nakamoto Consensus Game, parameterized by $\ell>0$ is identical to the Nakamoto Consensus Game \emph{except} in who is selected to mine during a step $t$. Instead of picking a miner proportional to their hashrate, a coin is flipped independently for each miner $i$ that is heads with probability $\ell \cdot \alpha_i$ (we will only ever consider games where $\ell \leq \max_j\{1/\alpha_j\}$). The set of coins is repeatedly flipped until at least one lands heads. Then, all miners whose coins are heads produce a block this round.

 As $\ell > 0$ approaches 0, the distribution of miners selected to create blocks in each round approaches that of the original Nakamoto Consensus Game, so we formally define 0-NCG as the original NCG (where exactly one miner is selected each round).

\nmgl\ is a natural extension of the original stylized model to incorporate a stylized model of latency. Even if every miner is honestly following the longest-chain protocol, there will inevitably be conflicts and orphaned blocks (for example, during any round in which there are multiple miners). We now build up language to formalize what we mean by statistical undetectability.

\begin{definition}[View of a Nakamoto Consensus Game] As a Nakamoto Consensus Game is played, we refer to the \emph{view} as the collection of all blocks that are eventually broadcast (treated as nodes in a directed graph --- the only information in the view is the pointer). We also refer to the \emph{view up to height $h$} as the collection of all blocks of height at most $h$ that are eventually broadcast. Observe that for any \nmgl\ and any strategy profile, the view is drawn according to some distribution.
\end{definition}

\begin{definition}[Statistically Undetectable Deviant Strategy] We say that a strategy $s$ for miner $i$ in \nmgl\ is \emph{$\ell'$-statistically undetectable} with respect to $\vec{s}_{-i}$ if the view of the game when Miner $i$ uses $s$ and and all other miners use $\vec{s}_{-i}$ is identically distributed to the view of the \nmglp\ when Miner $i$ uses some longest-chain strategy and all other miners use $\vec{s}_{-i}$.

We say that an ensemble of strategies $\{s^\varepsilon\}_{\varepsilon > 0}$ is \emph{statistically undetectable} with respect to $\vec{s}_{-i}$ if for all $\varepsilon > 0$ there exists an $\ell' \in (\ell,\ell+\varepsilon)$ such that $s^\varepsilon$ is $\ell'$-statistically undetectable with respect to $\vec{s}_{-i}$. 

In all applications of these definitions, each $s_j$ will be a longest-chain strategy (but may tie-break differently for different applications). 
\end{definition}

\noindent Note that in order to profit, Attacker must create orphans even if the base game is NCG. Therefore, in order for the view to look consistent with fully honest participants, Attacker will need to target an $\ell' > 0$ (and hence, even if one primarily wishes to study the original NCG as the ``real world'', one must study $\ell'> 0$ for the world created by Attacker).

Let us now quickly understand why both Selfish Mining Strategies are statistically detectable. Both claims will use the following notation (which our later proofs will also use) -- the purpose of this is to ease notational burden and support a reduction from analyzing the full $n$-player game to a $2$-player game. To ease notation here, and in the rest of the paper, we will w.l.o.g.~always consider Player $1$ to be the strategic player. 

\begin{definition}[Single/Pair] We say that height $h$ in a view of a Nakamoto Consensus Game has state Pair if there is a block of height $h$ created by Player $1$, and also a block of height $h$ created by a Player $>1$ (possibly multiple players). Otherwise, $h$ has state Single.
\end{definition}

\begin{definition}[SP-Simple] Consider any execution of a Nakamoto Consensus Game where all other players use the same deterministic longest-chain strategy. Consider also modifying the execution so that in every round, if at least one Miner $>1$ is selected to create a block, instead only Miner $2$ creates a block. A strategy for Miner $1$ is \emph{SP-Simple} if it takes identical actions in both games.

That is, when playing against a profile of identical deterministic longest-chain strategies, an SP-Simple strategy is agnostic to how many blocks are created by other miners during a particular round, or which miner created them, given that some other miner created a block during this round.
\end{definition}

\begin{proposition}\label{prop:SPOK} Let $s$ be any SP-Simple strategy for Miner $1$ and $\vec{s}_{-1}$ be a profile where all other players use the same longest-chain strategy $s'$, and $s'$ tiebreaks in either lexicographical or reverse lexicographical order.\footnote{That is, $s'$ either always tiebreaks in favor of $1$, or never in favor of $1$.} Then for all $\ell, \ell'$, $s$ is $\ell'$-statistically undetectable for \nmgl\ with hashrates $\vec{\alpha}$ with respect to $\vec{s}_{-i}$ if and only if $s$ is $\ell'$-statistically undetectable for \nmgl\ with hashrates $\langle \alpha_1,\frac{1-\prod_{i=2}^n (1-\alpha_i \cdot \ell))}{\ell}\rangle$ with respect to $s'$. 

Moreover, Miner $1$'s reward using $s$ in \nmgl\ with hashrates $\vec{\alpha}$ against $\vec{s}_{-i}$ is equal to that when using $s$ in \nmgl\ with hashrates $\langle \alpha_1,\frac{1-\prod_{i=2}^n (1-\alpha_i \cdot \ell))}{\ell}\rangle$ against $s'$.
\end{proposition}

A proof of Proposition~\ref{prop:SPOK} appears in Appendix~\ref{app:omitted}. Intuitively, Proposition~\ref{prop:SPOK} follows by observing that a single Player 2 produces a block in the proposed two-player game with the same probability that any player $>1$ produces a block in the original game. Proposition~\ref{prop:SPOK} allows us to focus just on the two-player game, which has significantly simpler notation.

\begin{proposition}\label{prop:nostrong} For any $\vec{\alpha}$, and any player $i$, Strong Selfish Mining is \emph{not} $\ell'$-statistically undetectable for any $\ell'$ in the Nakamoto Consensus Game with respect to the strategy profile where other players use longest-chain and tiebreak in favor of $i$.
\end{proposition}
\begin{proof} We use Proposition~\ref{prop:SPOK} and prove the claim for every two-player game. Indeed, observe that in any two-player game parameterized by $\ell$, if we let $\beta' = \frac{\alpha_1 \cdot \ell \cdot \alpha_2 \cdot \ell}{1-\alpha_1\cdot \ell \cdot \alpha_2 \cdot \ell}$ denote the probability that both players produce a block in the same round, then height $h$ is Pair with probability $\beta'$ independently for all $h$. 

If instead we consider a Nakamoto Consensus Game where Miner $1$ uses Strong Selfish Mining, then observe that height $h$ is Single if and only if Miner $2$ produces the first block at height $h$, and is Pair if and only if Miner $1$ produces the first block at height $h$.

So consider any height $h$ such that $h$ is Single. This means that Miner $2$ produced the first block at height $h$ and immediately broadcast it. Therefore, Miner $1$ produces the first block at height $h+1$ with probability $\alpha_1/(\alpha_1+\alpha_2)$, and height $h+1$ has state Pair with exactly this probability. If instead $h$ is Pair, then Miner $1$ has (at least) one withheld block, and Miner $1$ produces the first block of height $h+1$ with probability $>\alpha_1/(\alpha_1+\alpha_2)$ (because Miner $2$ would need to produce at least the next two blocks in order to produce the first block of height $h+1$). Therefore, the probability of seeing Single vs.~Pair at height $h+1$ depends on whether height $h$ is Single vs.~Pair, and this cannot be identically distributed to honest parties in a Nakamoto consensus game with any latency.
\end{proof}

Intuitively, Proposition~\ref{prop:nostrong} observes that Strong Selfish Mining is more likely to produce a Pair when it has just produced a Pair (because it must already have hidden blocks at the moment a height is determined to be Pair).

\begin{proposition} For any $\vec{\alpha}$, and any player $i$, Selfish Mining is \emph{not} $\ell'$-statistically undetectable for any $\ell'$ in the Nakamoto Consensus Game with respect to the strategy profile where other players use longest-chain and tiebreak against $i$.
\end{proposition}
\begin{proof} We use Proposition~\ref{prop:SPOK} and prove the claim for every two-player game. Indeed, observe that in any two-player game parameterized by $\ell$, if we let $\beta' = \frac{\alpha_1 \cdot \ell \cdot \alpha_2 \cdot \ell}{1-\alpha_1\cdot \ell \cdot \alpha_2 \cdot \ell}$ denote the probability that both players produce a block in the same round, then height $h$ is Pair with probability $\beta'$ independently for all $h$. 

If instead we consider a Nakamoto Consensus Game where Miner $1$ uses Selfish Mining, then observe that height $h$ is Single if and only if Miner $1$ has no hidden blocks when the block of height $h$ is broadcast. In particular, if two consecutive states $(h, h+1)$ are (Single, Pair), it must be that Miner $1$ has no hidden blocks at height $h$, and $h+1$ is determined to be state Pair as soon as Miner $1$ finds the next block. From here, state $h+2$ is Pair if and only if Miner $1$ finds the next block, which occurs with probability $\alpha_1/(\alpha_1+\alpha_2)$.

If instead two consecutive states $(h,h+1)$ are (Pair, Pair), it must be that Miner $1$ found a block at height $h+2$ before Miner $2$ found a block at height $h$, state $h+1$ is determined to be Pair the moment this happens. From here, state $h+2$ is Pair with probability $>\alpha_1/(\alpha_1+\alpha_2)$ (because Miner $2$ would need to produce at least the next two blocks in order for Miner $1$'s block at height $h+2$ to become pivotal). Therefore, the probability of seeing Single vs.~Pair depends on whether the previous two heights are (Single, Pair) vs.~(Pair, Pair), and this cannot be identically distributed to honest parties in a Nakamoto Consensus Game with any latency.
\end{proof}

Importantly, the above propositions note that Selfish Mining is statistically detectable \emph{only by looking at the view of the game}. In particular, it does not require timestamping information (which can easily be manipulated~\cite{YaishTZ22,YaishSZ22} and intentionally has minimal role in the consensus protocol), or real-time monitoring of the network (the entire purpose of a consensus protocol is to cope with the fact that messages arrive at unpredictable times due to latency). Still, we briefly discuss in Section~\ref{sec:conclusion} alternative detection methods based on these.

\section{Main Results and Technical Outline}\label{sec:results}
We now state our main results, which provide a statistically undetectable and profitable Selfish Mining strategy. All of our strategies have the following format: they begin with a strictly profitable selfish mining strategy (for our warmup, Strong Selfish Mining. For our first main result, Selfish Mining. For our second main result, a natural extension of Selfish Mining when there are sometimes natural forks), and sometimes ``wastes'' a hidden lead in order to preserve undetectability. That is, when a strategic miner has (say) five hidden blocks, they can reap the greatest profit by causing a fork with at least the first four. But, this creates a lot of sequential orphans. Our strategies will sometimes broadcast some of these blocks without creating an orphan. This sacrifices profit, but enables the strategy to appear indistinguishable from latency. The challenge is finding an appropriate broadcast pattern to be undetectable, and also analyzing when such patterns are still strictly 
profitable. We first begin with a warmup theorem that demonstrates some key ideas. 

\begin{theorem}\label{thm:warmup} Let $\vec{s}_{-1}$ be the strategy profile where every miner uses longest-chain and tiebreaks for Miner $1$. Then for any $\vec{\alpha}$, there is a statistically undetectable ensemble of strictly profitable strategies for Miner $1$ in the Nakamoto Consensus Game with respect to $\vec{s}_{-1}$.
\end{theorem}

\noindent We provide a complete proof of Theorem~\ref{thm:warmup} in Section~\ref{sec:warmup}. We give quantitative bounds on the profitability of the strategy in Lemma~\ref{lem:warmupprofitable}. We then prove our first main result:

\begin{theorem}\label{thm:main1} Let $\vec{s}_{-1}$ be the strategy profile where every miner uses longest-chain and tiebreaks against Miner $1$. Then for any $\vec{\alpha}$ with $\alpha_1 \geq \frac{3-\sqrt{5}}{2} \cdot \sum_j \alpha_j$,\footnote{Note that $\frac{3-\sqrt{5}}{2} \lessapprox 38.2\%$.} there is a statistically undetectable ensemble of strictly profitable strategies for Miner $1$ in the Nakamoto Consensus Game with respect to $\vec{s}_{-1}$.
\end{theorem}

Our bound of $\frac{3-\sqrt{5}}{2}$ is not tight for our strategy, although we provide a complete description of an infinite Markov chain that can be simulated in order to nail the tight bound to higher precision.\footnote{We suspect our bound is not far from tight, however.} Interestingly, even our loose bound is not too far from the $1/3$-fraction of hashrate that is required for Selfish Mining to be profitable without concern for undetectability. This establishes that adding undetectability to a profitable strategy (if possible) need not significantly detriment its performance. 

We include a complete proof of Theorem~\ref{thm:main1} in Section~\ref{sec:main1}. Quantitative bounds on the profitability of this strategy are given in Claim~\ref{claim:solopairs} and Lemma~\ref{lem:pair}. Finally, we extend Theorem~\ref{thm:main1} to begin from any parameterized Nakamoto Consensus Game.

\begin{theorem}\label{thm:main2} Let $\vec{s}_{-1}$ be the strategy profile where every miner uses longest-chain and tiebreaks against Miner $1$. Then for any $\vec{\alpha}$ with $\alpha_1 \geq 0.382 \cdot \sum_j \alpha_j$, there is a statistically undetectable ensemble of strictly profitable strategies for Miner $1$ in the Nakamoto Consensus Game parameterized by $\ell$ with respect to $\vec{s}_{-1}$.\footnote{Proposition~\ref{prop:maingeneral} provides a (significantly) more precise sufficient condition on $\alpha,\beta'$ in order for a strictly profitable and undetectable strategy to exist. Theorem~\ref{thm:main2} follows by verifying that this condition holds for all $\alpha \geq 0.382$ and all $\beta'\in [0,1]$. It could be that our sufficient condition also holds for all $\alpha \geq \frac{3-\sqrt{5}}{2}$ and all $\beta \in [0,1]$, but our computational software is not precise enough to verify this. Recall that $\frac{3-\sqrt{5}}{2} \lessapprox 0.382$, so the gap in analysis for $\beta'>0$ and $\beta'=0$ is fairly small.}

\end{theorem}

At the end of Section~\ref{sec:main1}, we overview ways in which the proof of Theorem~\ref{thm:main2} requires additional complexity, and defer a complete proof to Appendix~\ref{app:general}. In summary, Theorem~\ref{thm:main2} establishes that a strategic miner can strictly profit against honest miners who tiebreak against them in a way so that the view appears as though all miners are honest but with an arbitrarily small increase to latency.

\section{Warmup: Tiebreaking in favor of Strategic Miner}\label{sec:warmup}
In this section, we prove Theorem~\ref{thm:warmup}, which introduces some of the key ideas of our analysis. 

For simplicity of notation in this section and in all remaining technical sections, we let $n=2$ and $\alpha:=\alpha_1/(\alpha_1+\alpha_2)$. This is w.l.o.g.~due to Proposition~\ref{prop:SPOK}. We will also refer to Miner $1$ as Attacker and Miner $2$ as Honest. We first remind the reader of the concept of a state, specialized to $n=2$.

\begin{definition}[State]
    When $n=2$, observe that the \emph{state} of a height $h$, $S(h)$, is Pair if and only if both players create a block of height $h$, and Single if only one player creates a block of height $h$. We let $S(0)=\text{Single}$, namely there is a unique genesis block.
\end{definition}

We will describe our strategy first as labeling all heights as Single or Pair, and then describe the actual broadcasting strategy to match it. The outline of this section is as follows: Definition~\ref{def:warmuplabelling} describes the labeling strategy. Lemma~\ref{lem:warmupbroadcast-SP} specifies how the labeling strategy can be implemented via a broadcasting schedule. Lemma~\ref{lem:warmupundetectable-SP} argues that this strategy achieves undetectability. Finally, Lemma~\ref{lem:warmupprofitable} argues that the strategy is strictly profitable.

Recall that our goal is to produce a view so that each height is Pair with probability exactly $\beta$, for some $\beta:=\frac{\alpha \cdot \ell'\cdot (1-\alpha) \cdot \ell'}{1-(1-\alpha )\cdot \ell'\cdot \alpha \cdot \ell'}$,\footnote{This computation of $\beta$ follows by observing that in a two-player \nmglp, both players produce a block during the same timestep with probability $\frac{\alpha \cdot \ell'\cdot (1-\alpha) \cdot \ell'}{1-(1-\alpha )\cdot \ell'\cdot \alpha \cdot \ell'}$, independently.} and some $\ell'$ arbitrarily close to $0$ (which corresponds to $\beta$ arbitrarily close to $0$ as well).

\begin{definition}[Labeling Strategy]\label{def:warmuplabelling} Define $P_h$ to be the probability that Attacker creates the first block of height $h$, conditioned on all information available as of the moment that the first block of height $h-1$ is created. Then:
\begin{itemize}[nosep]
\item If Honest creates the first block at height $h$, label $h$ as Single.
\item If Attacker creates the first block at height $h$, label $h$ as Pair with probability $\beta/P_h$, and Single otherwise. 
\end{itemize}
\end{definition}

Observe that Attacker can easily compute $P_h$ as a function of the number blocks it is hiding: If Attacker has $i$ hidden blocks at the moment the first block of height $h-1$ is created, Attacker will mine the first block of height $h$ unless Honest mines $i+1$ blocks in a row. That is, $P_h=1-(1-\alpha)^{i+1}$. Observe also that this expression is at least $\alpha$ for all $i\geq 0$, so the labelling strategy is valid (that is, $\beta/P_h$ is a valid probability) as long as $\beta \leq \alpha$. 




Next, we need to define how to implement the proposed labeling strategy with an actual broadcasting strategy in the game. In particular, we need to specify a broadcasting schedule that realizes the labels output by the labeling strategy.

\begin{lemma}[Implementability]\label{lem:warmupbroadcast-SP}
    The following \emph{broadcasting strategy} realizes the labels output by the labeling strategy in Definition~\ref{def:warmuplabelling}:
        \begin{itemize}[nosep]
            \item If $S(h)$ is labeled as Single (and Attacker produced a block of height $h$): broadcast the block of height $h$ the moment a block of height $h-1$ is broadcast (which might mean broadcasting the moment that Attacker produces a block of height $h$).
            \item If $S(h)$ is labeled as Pair: broadcast the block of height $h$ the moment an Honest block of height $h$ is published.
        \end{itemize}
\end{lemma}
\begin{proof}
    Clearly, if this can be implemented, it will cause the state of every block to match its label (because there will clearly not be a chance for Honest to conflict with the Single broadcasts, and Honest clearly will conflict with the Pairs). 
    
    Moreover, this strategy can be implemented: the state of $h$ is determined the moment a block of height $h$ is created. If you create a block of height $h$, you learn its state immediately, and can either broadcast it immediately (if labeled Single), or broadcast once contested (if labeled Pair).
\end{proof}

\begin{lemma}[Undetectability]\label{lem:warmupundetectable-SP} The broadcasting strategy of Lemma~\ref{lem:warmupbroadcast-SP} applied to the labeling strategy of Definition~\ref{def:warmuplabelling}, for any $\beta \leq \alpha$, played against a longest-chain strategy that breaks ties in favor of Attacker, produces a view where each height has state Pair independently with probability $\beta$.
\end{lemma}
\begin{proof}
    Observe that at the moment that a block of height $h-1$ is first created, $S(h)$ is equal to Pair with probability exactly $\beta$. This is true not only conditioned on $S(1),\ldots, S(h-1)$, but also conditioned on any additional information known at the instant $S(h-1)$ is fixed. Therefore, for all $S(1),\ldots, S(h-1)$, the probability that $S(h)$ is Pair conditioned on $S(1),\ldots, S(h-1)$ is exactly $\beta$. Therefore, the distribution of states is independent each round, and equal to Pair w.p.~$\beta$.
\end{proof}
\noindent
Lemma~\ref{lem:warmupundetectable-SP} concludes undetectability, which holds for $\beta$ arbitrarily close to $0$ (and therefore $\ell'$ arbitrarily close to $0$ as well). 

\begin{lemma}[Profitability]\label{lem:warmupprofitable}
The broadcasting strategy of Lemma~\ref{lem:warmupbroadcast-SP} applied to the labeling strategy of Definition~\ref{def:warmuplabelling}, for any $\beta \leq \alpha$, when played against a longest-chain strategy that breaks ties in favor of Attacker, achieves reward $\alpha + \alpha \beta > \alpha$.
\end{lemma}
\begin{proof}
The proof is a direct application of Lemma~\ref{lem:pair} (stated and proved immediately below in Section~\ref{sec:technical}), observing that the strategy wins all Pairs.
\end{proof}

\begin{proof}[Proof of Theorem~\ref{thm:warmup}]
The proof follows from Proposition~\ref{prop:SPOK}, Lemma~\ref{lem:warmupundetectable-SP}, and Lemma~\ref{lem:warmupprofitable}.
\end{proof}

\subsection{Technical Lemmas}\label{sec:technical}
Below are two technical lemmas that we use in this section and will reuse in later sections as well. 

\begin{lemma}\label{lem:pair} Let $n=2$, and consider any strategy that eventually broadcasts all blocks when playing against a longest-chain. Then, if this strategy wins a $\delta$ fraction of Pairs, and Pairs occur in a $\beta$ fraction of rounds, it achieves expected reward:
$$\alpha - (1-\alpha-\delta)\cdot \beta.$$
\end{lemma}
\begin{proof}
Say that the strategy results in a $\beta$ fraction of rounds having Pairs. Observe that an $\alpha$ fraction of all blocks are created by Attacker, and a $(1-\alpha)$ fraction of all blocks are Honest. Observe further that every Single round has exactly one block, and every Pair round has one block from each creator. Therefore, we see that $\alpha\cdot(1+\beta)$ blocks must be produced by Attacker per-round. Moreover, $\beta$ blocks per round of Attacker are in Pair rounds, leaving $\alpha \cdot (1+\beta) - \beta$ blocks per round that are Attacker in Single rounds. This means that Attacker wins $\alpha \cdot (1+\beta) - \beta$ fraction of rounds because they are Attacker Single, and an additional $\delta \cdot \beta$ fraction of rounds because they win $\delta$ fraction of Pairs.
\end{proof}

\begin{corollary}\label{cor:paircounting} An SP-Simple strategy that eventually broadcasts all blocks and wins a $\delta$ fraction of Pairs achieves reward strictly better than $\alpha$ if and only if $\delta > 1-\alpha$.
\end{corollary}
\begin{proof} Observe that $\alpha  - (1-\alpha - \delta)\cdot \beta >\alpha \Leftrightarrow -(1-\alpha - \delta) > 0 \Leftrightarrow \delta > 1-\alpha$.
\end{proof}

\section{Main Result I: Tiebreaking against Strategic Miner}\label{sec:main1}
Here, we prove our first main result: a profitable and statistically undetectable Selfish Mining strategy when all miners tiebreak against the strategic miner. In the same canonical model where Selfish Mining was first introduced, we show that a modified strategy is strictly profitable and also statistically undetectable. We begin with some concepts, and will also reuse concepts/definitions from our warmup.

\subsection{Concepts}

\begin{definition}[Pivotal Block]
    A block of height $h$ created by Attacker is \emph{pivotal} if height $S(h-1)$ is Pair, and when the honest block of height $h-1$ is broadcast, Attacker does not have a block of height $h+1$. That is, Attacker's block of height $h$ is pivotal if there is a conflict at height $h-1$, and Attacker needs to use this block in order to win it.
\end{definition}

\begin{definition}[Safe Block]
    A block created by Attacker is \emph{safe} if it is not pivotal. Observe that a block of height $h$ can be safe if either $S(h-1)$ is single, or if Attacker finds a block of height $h+1$ before Honest finds a block of height $h-1$.
\end{definition}

Intuitively, just like Selfish Mining, we really want to broadcast all pivotal blocks immediately, so that we don't risk losing a conflict that we can win right now. If a block is safe, we don't need to broadcast it immediately, but our strategy may choose to do so anyway to maintain undetectability.

\subsection{The Strategy}
We again first describe the strategy as labeling all heights as Single or Pair, and then describe the actual broadcasting strategy to match it. 

\begin{definition}[Labeling Strategy]\label{def:labelling}
    Define $P_h$ to be the probability that Attacker creates a block of height $h$ and that block is safe, conditioned on all information available \emph{as of the first moment we know} $S(h')$ for all $h' < h$.\footnote{The value of $P_h$ is easily computable; see Table~\ref{tab:Ph} for explicit expressions.} The labeling is strategy is then as follows:
    \begin{itemize}[nosep]
        \item If Honest first creates a block at height $h$, then $h$ is labeled Single.
        \item If Attacker first creates a block at height $h$, and it is pivotal, then $h$ is labeled Single.
        \item If Attacker first creates a block at height $h$, and it is safe, then $h$ is labeled Pair with probability $\beta/P_h$ and Single otherwise.
    \end{itemize}
    In particular, the instant that we have the necessary information to label a height, we label it (and the instant we have all the necessary information to know whether to flip a coin, we flip it).
\end{definition}

Lemma~\ref{lem:validity} will ensure that the coin flip probabilities in the third bullet are valid. Before that, we will first show the following useful claim about the moment when there is enough information to decide the label of a given height. Importantly, we \textit{don't} necessarily know $S(h)$ the moment its created (unlike in the warmup, where this information was sufficient), so we have to carefully track when it is first determined, and the information available at that time.

\begin{claim}[When labels are determined]\label{claim:when-label}
    For all $h$, there is sufficient information to determine $S(h)$ by the moment that a block of height $h+1$ is created (possibly sooner).
\end{claim}
\begin{proof}
    We'll proceed by induction on $h$. The base case holds vacuously when $h=0$, since the genesis block is Single by definition.
    
    Let's first consider the moment that the block of height $h$ is created. If the first block of height $h$ is created by Honest, then $h$ is certainly Single (and we learn this immediately when the first block of height $h$ is created, which is before the first block of height $h+1$ is created). 

    If instead, the first block of height $h$ is created by Attacker, then there are a few cases to consider. Recall that by the inductive hypothesis, we certainly know the labels of all blocks of height $h' < h$. This in particular implies that we have all the necessary information to compute $P_h$, and can already flip any coins that might be needed to label $h$. The only remaining uncertainty is around whether $h$ is safe or pivotal. Zooming in on $S(h-1)$:
    \begin{enumerate}[label=(\alph*),nosep]
        \item If $S(h-1)$ is Single, then $h$ is safe (by definition), and we learn this the moment that the first block at height $h$ is created.
        \item If $S(h-1)$ is Pair, and Honest has broadcast a block of height $h-1$, then $h$ is pivotal, and we learn this the moment $h$ is created, so we can label $S(h)$ as Single.
        \item If $S(h-1)$ is Pair, and Honest has \emph{not} yet broadcast a block of height $h-1$, then we do not yet know whether $h$ is safe or pivotal. We learn this the moment either of the following events happens (whichever comes first):
        \begin{itemize}
            \item Attacker finds another block (which will be at height $h+1$); then $h$ is safe.
            \item Honest finds a block at height $h-1$ (which may or may not require Honest to find multiple blocks); then $h$ is pivotal.
        \end{itemize}
        In particular, by the time there is a block of height $h+1$, we certainly know which case occurred first, and thus know whether $h$ is safe or pivotal. From here, we can flip the necessary coins to label $S(h)$. \qedhere
    \end{enumerate}
\end{proof}

\begin{corollary}
    For all $h$, $P_h$ can be evaluated by the time the first block of height $h$ is created.
\end{corollary}

\begin{lemma}[Validity]\label{lem:validity}
    The probabilities in Definition~\ref{def:labelling} are valid for any $\beta\leq \alpha^2$.
\end{lemma}
\begin{proof}
    We will show that for all $h$, we have $P_{h+1} \geq \alpha^2$, which completes the proof.

    Let us again focus on the moment a block of height $h$ is created for the first time. Consider the bullets in the proof of Claim~\ref{claim:when-label}. In bullets (a) and (b), we learn $S(h)$ the moment that the block at height $h$ is created. Then if Attacker creates the next two blocks, they will have created a safe block at height $h+1$. Therefore, $P_{h+1} \geq \alpha^2$ in both of these cases. 
    \\ \noindent
    If we learn $S(h)$ through bullet (c), then:
    \begin{itemize}[nosep]
        \item If it is because Attacker found a block at height $h+1$, then if they find the next block, this block of height $h+1$ certainly safe (because Attacker found a block of height $h+1$ before Honest found a block of height $h-1$). So $P_{h+1} \geq \alpha \geq \alpha^2$.
        \item If it is because Honest creates a block at height $h-1$, then Attacker broadcasts their pivotal block of height $h$ and this acts like a new genesis block. From here, if Attacker creates the next two blocks, they will create a safe block at height $h+1$. So $P_{h+1} \geq \alpha^2$.\qedhere
    \end{itemize}
\end{proof}

Observe some subtlety in the proof of Lemma~\ref{lem:validity}, and in particular that the choice of moments to compute $P_{h}$ is significant. For example, if we were to compute the probability that Attacker creates a safe block at height $h$ conditioned on $S(h-2)$ being Single, and $S(h-1)$ being Pair, \emph{and Honest has broadcast a block of height $h-1$}, then the Attacker cannot possibly find a safe block of height $h$ (because their block of height $h$, if created, is immediately pivotal). So, it is important that we compute $P_h$ at the moment that $S(h-1)$ is determined (which is immediately when Attacker finds its first block at height $h-1$, and it is unknown whether Honest will find the next block or not). This highlights some of the subtlety needed to keep the analysis clean.

Next, we need to define how to implement the proposed labeling with a broadcasting strategy.

\begin{lemma}[Implementability]\label{lem:broadcast-SP}
    The following \emph{broadcasting strategy} realizes the labels output by the labeling strategy in Definition~\ref{def:labelling}:
        \begin{itemize}[nosep]
            \item If $S(h)$ is labeled as Single (and you created a block of height $h$): broadcast the block of height $h$ the moment a block of height $h-1$ is broadcast (and you know that $S(h)$ is labeled Single). Note that this may imply broadcasting a block of height $h$ the moment its created.
            \item If $S(h)$ is labeled as Pair: broadcast the block of height $h$ the moment another block of height $h$ is broadcast (and you know that $S(h)$ is labeled Pair).\footnote{In particular, if the public longest chain has height $h-2$, and Attacker has two hidden blocks of heights $h-1$ labeled Pair and $h$ labeled Single, when Honest broadcasts a block of height $h-1$, both bullets trigger simultaneously, and the attacker will broadcast both its hidden blocks at once.}
        \end{itemize}
\end{lemma}
\begin{proof}

We first need to confirm that for all $h$ that are eventually labeled as Single, we learn that $S(h)$ is Single early enough to broadcast our block of height $h$ according to the rule above \emph{before Honest finds a block of height $h$}. If we can do this, then $h$ will indeed be Single. 
    
    To see this, observe that if a block of height $h$ is pivotal, we learn this immediately when there are two blocks of height $h-1$, and Attacker creates the first block of height $h$. This is before Honest has a block of height $h$, and therefore if a block is pivotal, we label it as Single in time.
    
    Similarly, if a block of height $h$ is safe, we learn this either the moment we learn that $S(h-1)$ is Single (which we certainly know by the moment the block of height $h$ is created, perhaps earlier, by Claim~\ref{claim:when-label}), or the moment that we create a block of height $h+1$ before Honest has a block of height $h-1$. In both cases, Honest does not have a block of height $h$. This confirms that if we decide $S(h)$ is Single, we know this before Honest finds a block of height $h$, and therefore our broadcasting strategy broadcasts our block of height $h$ in time for $h$ to be Single.

    For broadcasting Pairs, we just need to confirm that the timing constraints do not conflict with those demanded for Single (that is, we need to confirm that the strategy does not accidentally claim to broadcast a block of height $h+1$ before a block of height $h$ that it points to). Assuming this is the case, then because we wait until Honest broadcasts their block of height $h$ before broadcasting ours, there will certainly be a Pair at height $h$. To see that there is no conflict, observe that the broadcasting strategy for $h$ labeled Pair only asks to wait to broadcast a block of height $h$ until Honest finds a block of height $h$, whereas for $h+1$ labeled Single it only asks to broadcast a block of height $h+1$ as soon as a block of height $h$ is broadcast. Therefore, there is no conflict where the broadcasting strategy asks to broadcast a block of height $h+1$ before a block of height $h$.
\end{proof}

\begin{lemma}[Undetectability]\label{lem:undetectable-SP} The broadcasting strategy of Lemma~\ref{lem:broadcast-SP}, with the labeling strategy of Definition~\ref{def:labelling}, for any $\beta \leq \alpha^2$, produces the distribution where a $\beta$-fraction of rounds have Pairs, and $1-\beta$ have Singles, independently.
\end{lemma}
\begin{proof}
    Observe that at the moment that $S(1),\ldots, S(h-1)$ are fixed, $S(h)$ is equal to Pair with probability exactly $\beta$. This is true not only conditioned on $S(1),\ldots, S(h-1)$, but also conditioned on any additional information known at the instant $S(h-1)$ is fixed. Therefore, for all $S(1),\ldots, S(h-1)$, the probability that $S(h)$ is Pair conditioned on $S(1),\ldots, S(h-1)$ is exactly $\beta$. Therefore, the distribution of states is independent each round, and equal to Pair with probability $\beta$.
\end{proof}

\subsection{Reward Analysis}
We provide in Appendix~\ref{app:exactreward} a description of a Markov Chain that can be used to compute exactly the reward of our strategy. In this section, we'll aim instead to simply find a sufficiently large $\alpha$ so that our strategy is strictly profitable.

Recall that, by Corollary~\ref{cor:paircounting}, it suffices to focus on Paired rounds and determine their winner. Consider the sequence $(S(h))_{h=1}^\infty$, and observe that by Lemma~\ref{lem:undetectable-SP}, each of its entries is Pair with probability $\beta$ and Single otherwise. Consider now a sequence of $k$ Pairs between two Single rounds.
\begin{claim} For a sequence of $k> 1$ consecutive Pair rounds, Attacker wins all these rounds.
\end{claim}
\begin{proof} Say that height $h$ is the first Pair. Observe that in order for each subsequent height to possibly be labeled as Pair, Attacker must have a safe block at that height. In order to have a safe block at heights $h,\ldots, h+k-1$, Attacker must also have a block at height $h+k$ before Honest has a block of height $h+k-1$, and height $h+k$ is Single. Therefore, Attacker will certainly win all these Pairs.
\end{proof}
\noindent
We can therefore focus on \emph{Solo Pairs}: $h$ such that $S(h)$ is Pair, and $S(h-1)$ and $S(h+1)$ are Single.

\begin{claim}\label{claim:solopairs} Let $w$  denote the fraction of Solo Pairs won by the attacker. We have $w\geq \frac{2\alpha - \alpha^2 -\beta}{1-\beta}$.
\end{claim}
\begin{proof}
Consider a Single at height $h-1$, and a Pair at height $h$. If Attacker has any hidden blocks, then Attacker will certainly win the pair at height $h$. Otherwise, $P_{h+1} = \alpha^2$ (because Attacker's next block is safe only if they create each of the next two). This means that:
\begin{itemize}[nosep]
\item If Honest finds the next two blocks, then $h$ is a Solo Pair, and Honest wins it. This happens with probability $(1-\alpha)^2$.
\item If Attacker finds exactly one of the next two blocks, then their block of height $h+1$ is pivotal, so $h$ is a Solo Pair and Attacker wins it. This happens with probability $2\alpha(1-\alpha)$.
\item If Attacker finds both of the next blocks, then $h+1$ is safe, but is marked Single anyway with probability $1-\beta/\alpha^2$. Therefore, this is a Solo Pair that Attacker wins with probability $\alpha^2\cdot (1-\beta/\alpha^2) = \alpha^2 - \beta$.
\end{itemize}
\noindent
Altogether, we conclude that Attacker wins at least a $\frac{2\alpha - \alpha^2 -\beta}{1-\beta}$ fraction of Solo Pairs.
\end{proof}

\begin{corollary} The fraction of pairs won by Attacker is $>(1-\alpha)$ as long as $\alpha > \frac{3-2\beta -\sqrt{5-4\beta}}{2(1-\beta)}$.
\end{corollary}
\begin{proof}
Recall that the number of Pair rounds in a row (between two Single rounds) is distributed independently across time, and equal to $i$ with probability $(1-\beta)\cdot \beta^i$.
Therefore, we can write the expected number of Pair rounds between two Single rounds as
\begin{align*}
    \sum_{i=0}^\infty (1-\beta)\cdot \beta^i \cdot i = \frac{\beta}{1-\beta}.
\end{align*}
Similarly, we can write an expression for the expected number of Pair rounds won by the attacker,
\begin{align*}
    w\cdot (1-\beta)\beta + \sum_{i=2}^\infty (1-\beta)\cdot \beta^i \cdot i =
    w\cdot\beta\cdot(1-\beta) + \frac{(2-\beta)\cdot \beta^2}{1-\beta} = \frac{w\cdot \beta \cdot (1-\beta)^2 + (2-\beta)\cdot \beta^2}{1-\beta},
\end{align*}
where $w$ denotes the fraction of Solo Pairs won.

Now, we can compute the expected fraction of Pairs won by the attacker. By linearity of expectation, the expected number of Pair rounds before we see the $n^{th}$ Single is $n\cdot \frac{\beta}{1-\beta}$, and the expected number of Pair rounds won by the attacker is $n \cdot \frac{w \beta (1-\beta)^2 + (2-\beta)\beta^2}{1-\beta}$. By the law of large numbers, we have that with probability $1$:

$$\lim_{n\rightarrow \infty}(\text{\# Pairs before }n^{th}\text{ Single})/n = \frac{\beta}{1-\beta},$$
$$\lim_{n\rightarrow \infty}(\text{\# Pairs won by Attacker before }n^{th}\text{ Single})/n = \frac{w\cdot \beta \cdot (1-\beta)^2 + (2-\beta)\cdot \beta^2}{1-\beta}.$$

Therefore, the expected fraction of Pairs won by the attacker is:

$$\frac{\frac{(2-\beta)\cdot \beta^2 + w\cdot \beta \cdot (1-\beta)^2}{1-\beta}}{\frac{\beta}{1-\beta}} = \frac{w\cdot \beta \cdot (1-\beta)^2 + (2-\beta)\cdot \beta^2}{1-\beta}.$$

Substituting in our lower bound on $w$ from Claim~\ref{claim:solopairs} in the right hand side gives the following lower bound on the expected fraction of Pairs won by the attacker:
\begin{align*}
2\beta - \beta^2 + w \cdot (1-\beta)^2 &\geq 2\beta - \beta^2 + (2\alpha-\alpha^2 - \beta)\cdot (1-\beta).
\end{align*}

And now we need to understand when $2\beta - \beta^2 + (2\alpha-\alpha^2 - \beta)\cdot (1-\beta) \geq 1-\alpha$. Rearranging, this inequality is equivalent to $\alpha^2 \cdot (\beta-1) + \alpha \cdot (3-2\beta) + (\beta-1) \geq 0$. The left-hand-side has roots of: $$\frac{-3+2\beta \pm\sqrt{(3-2\beta)^2-4(1-\beta)^2}}{2(\beta-1)}=\frac{3-2\beta \pm\sqrt{5-4\beta}}{2(1-\beta)}.$$
So the inequality holds iff
$$\alpha \in \left(\frac{3-2\beta - \sqrt{5-4\beta}}{2(1-\beta)}, \frac{3-2\beta +\sqrt{5-4\beta}}{2(1-\beta)}\right).$$
In particular, observe that the upper bound is at least $1$ for all $\beta \in [0,1]$ (the numerator is always at least $2$ and the denominator is at most $2$). So the desired claim holds as long as $\alpha > \frac{3-2\beta -\sqrt{5-4\beta}}{2(1-\beta)}$. \end{proof}

\begin{proof}[Proof of Theorem~\ref{thm:main1}]

Observe that $\frac{3-2\beta -\sqrt{5-4\beta}}{2(1-\beta)}$ is monotone decreasing in $\beta$.\footnote{Its derivative with respect to $\beta$ is $\frac{2\beta-3+\sqrt{5-4\beta}}{2 \cdot \sqrt{5-4\beta}\cdot (1-\beta)^2}$. The denominator is $\geq 0$ for all $\beta \in (0,1]$. The numerator is $0$ at $\beta = 1$, and has derivative $2-\frac{2}{\sqrt{4-5\beta}}$, which is positive on $(0,1)$. Therefore, the numerator is negative on the entire interval $(0,1)$. } Therefore, when $\alpha > \frac{3-\sqrt{5}}{2}$, we have $\alpha >\frac{3-2\beta -\sqrt{5-4\beta}}{2(1-\beta)}$ for all $\beta \in (0,1)$. Therefore, when $\alpha > \frac{3-\sqrt{5}}{2}$, there is a range of $\beta$ arbitrarily close to $0$ (and therefore a range of $\ell'$ arbitrarily close to $0$) where our strategy is both statistically undetectable and strictly profitable.
\end{proof}

We also leverage Corollary~\ref{cor:paircounting} to derive a bound on $\alpha$ that suffices for a strictly profitable strategy that is $\ell'$-statistically undetectable, but not necessarily $\ell'$ close to $0$.

\begin{proposition} Let $\alpha^*\approx 0.3586$ denote the unique real root in $(0,1)$ of $\alpha^4 - 2\alpha^3 +3\alpha-1$. Then for any $\alpha > \alpha^*$, there exists an $\ell'$ such that there is a strictly profitable strategy for the Nakamoto Consensus Game that is $\ell'$-statistically undetectable.
\end{proposition}
\begin{proof}
By Lemma~\ref{lem:undetectable-SP}, our proposed strategy is statistically undetectable for all $\beta \leq \alpha^2$, so we will take $\beta = \alpha^2$.\footnote{We have previously observed in the proof of Theorem~\ref{thm:main1} that the profitability threshold is monotone decreasing in $\beta$, so taking the largest viable $\beta$ will optimize this.} By Corollary~\ref{cor:paircounting}, we have designed a strategy that is strictly profitable (and produces the desired view) as long as $\alpha > \frac{3-2\alpha^2 -\sqrt{5-4\alpha^2}}{2(1-\alpha^2)}$. Expanding out the calculations, we need:
\begin{align*}
\alpha &>\frac{3-2\alpha^2 -\sqrt{5-4\alpha^2}}{2(1-\alpha^2)}\\
\Leftrightarrow 2\alpha(1-\alpha^2) &>3-2\alpha^2 -\sqrt{5-4\alpha^2}\\
\Leftrightarrow  3-2\alpha(1-\alpha^2)-2\alpha^2&< \sqrt{5-4\alpha^2}\\
\Leftrightarrow 3 - 2\alpha + 2\alpha^3- 2\alpha^2 &< \sqrt{5-4\alpha^2}
\end{align*}
Because we are only interested in the range $\alpha \in (0,1/2)$, the LHS and RHS are always non-negative. Therefore we may continue with:
\begin{align*}
3-2\alpha-2\alpha^2 +2\alpha^3&< \sqrt{5-4\alpha^2}\\
\Leftrightarrow (3-2\alpha-2\alpha^2 + 2\alpha^3)^2 &< 5-4\alpha^2\\
\Leftrightarrow 4\alpha^6 -8\alpha^5 -4\alpha^4 +20\alpha^3 -8\alpha^2 -12\alpha +9 -5 + 4\alpha^2 &< 0\\
\Leftrightarrow 4\alpha^6 -8\alpha^5 -4\alpha^4 +20\alpha^3 -4\alpha^2 -12\alpha +4 &< 0\\
\Leftrightarrow 4(\alpha-1)\cdot (1+\alpha)\cdot (\alpha^4 - 2\alpha^3 + 3\alpha -1) &< 0.
\end{align*}
The quartic on the LHS has two real roots. One is negative, and the other is $\approx 35.86\%$. Therefore, the quartic is positive on $[.3586,1]$, and the entire LHS is negative.
\end{proof}

\subsection{Looking Forward: Generalizing to Main Result II}
Here, we briefly note complexities extending to our second main result. Recall that our second main result considers statistically undetectable Selfish Mining strategies for an \nmgl\ with $\ell > 0$. 

The key difference stems from the following. In our warmup result, Attacker is relatively free to decide which heights will be Single vs.~Pair, because they will win every Pair round anyway. In our first main result, Attacker is sometimes ``forced'' to set an upcoming height to be Single, because they have only a pivotal block left (and having a conflict with this block is risky, and could cause the loss of many other conflicts that Attacker could instead win right now). In our second main result, Attacker is still sometimes ``forced'' to immediately broadcast pivotal blocks and create Single heights, but they are also sometimes ``forced'' to label a height as Pair (simply because both miners find a block at the same time). We now briefly elaborate how this impacts our analysis.

Because of these ``Forced Pairs'', our labeling strategy must change. One particular challenge is that we cannot set $P_h$ conditioned on the entire state of information known to the Attacker at a given time, as doing so prevents any dishonest strategy from also being statistically undetectable.\footnote{We refer the reader to the technical appendices to see why. Briefly, the issue is that if the attacker finds a single block following a ``Forced Pair'', that block is immediately pivotal, and therefore should be broadcast. But this then means that following a Forced Pair, we can only get another Pair if we get another Forced Pair, which happens with a fixed probability set by $\ell'$ that is outside our control.} Instead, our labeling strategy instead conditions only on the states, and \emph{not} on the full information available to Attacker. For example, our strategy computes (e.g.) $P_5$ (and other related probabilities) conditioned on the fact that the first four states are $\langle$Single, Single, Pair, Pair$\rangle$, and at the moment these states are set. But, it does \emph{not} condition on further information available to Attacker (such as whether the last two Pairs were Forced, whether Attacker has any hidden blocks, etc.). This forces our analysis to be Bayesian, and further understand the probabilities of a particular ``complete world'' given the current sequence of states. This is the key complexity associated with confirming statistical undetectability.

The key complexity associated with our reward calculation is the following. For our first main result, the honest strategy gets reward $\alpha$, and this provides a clean target fraction of pairs to win (of $>(1-\alpha)$) in order to be strictly profitable. In an \nmgl\ with $\ell > 0$, the honest strategy does not win an $\alpha$ fraction of rewards, so our benchmark is different and calculations become significantly more involved. In order to keep the calculations as tractable as possible, we use an analysis technique introduced in~\cite{SapirshteinSZ16} to directly understand when withholding Attacker's first block (and attempting to Selfish Mine) might be more profitable than immediately broadcasting it (and claiming immediate reward). In particular, Proposition~\ref{prop:maingeneral} provides a sufficient condition (as a function of $\alpha,\beta'$) for a strictly profitable and undetectable strategy to exist (and Theorem~\ref{thm:main2} follows by confirming numerically that this condition holds for $\alpha \geq 38.2\%$ and all $\beta'\in [0,1]$).

\section{Conclusion}\label{sec:conclusion}
We provide a statistically undetectable Selfish Mining attack. We choose to study the canonical setting of proof-of-work longest-chain protocols with a block reward, as it is the theoretically most-developed. The key takeaway from our work is that relying on risk of statistical detection to dissuade attackers is not necessarily sound --- there may be ways for attackers to still profit and avoid statistical detection. Our paper leaves several directions open for future work and discussion.

On the technical front, our paper considers a particular stylized model of latency, where orphans naturally occur independently every round. It is worthwhile to further explore alternative latency models, and develop an understanding of which types of ``honest-with-latency'' worlds can be induced by profitable deviations.

On the modeling front, our paper considers detection methods based only on the shape of the blockchain (e.g. the pattern of orphaned blocks), and not on any time-sensitive information (such as timestamps and timing of messages). On one hand, it is likely \emph{quite} challenging (and perhaps even impossible) to design a profitable Selfish Mining strategy whose timestamps mimic that of natural latency. On the other hand, if timestamps are necessary to detect a deviant strategy, and detection of Selfish Mining via timestamps could significantly impact the value of the underlying cryptocurrency, this is potentially a significant vulnerability (because timestamps are trivial to manipulate). Indeed, the deviation discovered by~\cite{YaishSZ22} involves manipulating timestamps on Ethereum (for the purpose of extracting additional mining revenues, not to affect the underlying value of ETH).

Also on the modeling front, our strategy doesn't risk significantly impacting the underlying cryptocurrency's value, as its impact is consistent with orphans caused by latency. Still, note that while Attacker can safely disguise its behavior as latency, our Attacker still introduces extra latency into the system (and in particular, causes orphaned blocks at a higher rate). It is entirely possible that an increased orphan rate might have a smooth negative impact on the underlying cryptocurrency's value without outright tanking it. This aspect would also be important to explore, although the model might look significantly different than ours.

On the broader web3/blockchain ecosystem front, the conceptual contributions of our paper apply quite generally: deviant behavior can sometimes disguised as natural occurrences. Therefore, our work motivates further study of any domain where deviant strategies exist that are profitable when denoted in the underlying token/cryptocurrency, but are currently believed to be disincentivized due to risk of detection. In general, wherever detectable profitable deviations exist, it is important to understand whether undetectable (or ``explainable-by-nature'') profitable deviations exist as well.

\bibliographystyle{plain}
\bibliography{MasterBib}

\begin{thebibliography}{10}

\bibitem{bahrani2023-TFM-MEV}
Maryam Bahrani, Pranav Garimidi, and Tim Roughgarden.
\newblock Transaction fee mechanism design with active block producers.
\newblock {\em arXiv preprint arXiv:2307.01686}, 2023.

\bibitem{BirmpasKLC20}
Georgios Birmpas, Elias Koutsoupias, Philip Lazos, and Francisco J.~Marmolejo Coss{\'{\i}}o.
\newblock Fairness and efficiency in dag-based cryptocurrencies.
\newblock In Joseph Bonneau and Nadia Heninger, editors, {\em Financial Cryptography and Data Security - 24th International Conference, {FC} 2020, Kota Kinabalu, Malaysia, February 10-14, 2020 Revised Selected Papers}, volume 12059 of {\em Lecture Notes in Computer Science}, pages 79--96. Springer, 2020.

\bibitem{BrownCohenNPW19}
Jonah Brown{-}Cohen, Arvind Narayanan, Alexandros Psomas, and S.~Matthew Weinberg.
\newblock Formal barriers to longest-chain proof-of-stake protocols.
\newblock In {\em Proceedings of the 2019 {ACM} Conference on Economics and Computation, {EC} 2019, Phoenix, AZ, USA, June 24-28, 2019.}, pages 459--473, 2019.

\bibitem{CarlstenKWN16}
Miles Carlsten, Harry~A. Kalodner, S.~Matthew Weinberg, and Arvind Narayanan.
\newblock On the instability of bitcoin without the block reward.
\newblock In {\em Proceedings of the 2016 {ACM} {SIGSAC} Conference on Computer and Communications Security, Vienna, Austria, October 24-28, 2016}, pages 154--167, 2016.

\bibitem{ChicarinoAJR20}
Vanessa Chicarino, C\'{e}lio Albuquerque, Emanuel Jesus, and Ant\^{o}nio Rocha.
\newblock On the detection of selfish mining and stalker attacks in blockchain networks.
\newblock {\em Annals of Telecommunications}, 75:143--152, 2020.

\bibitem{ChungS23}
Hao Chung and Elaine Shi.
\newblock Foundations of transaction fee mechanism design.
\newblock In Nikhil Bansal and Viswanath Nagarajan, editors, {\em Proceedings of the 2023 {ACM-SIAM} Symposium on Discrete Algorithms, {SODA} 2023, Florence, Italy, January 22-25, 2023}, pages 3856--3899. {SIAM}, 2023.

\bibitem{EyalS14}
Ittay Eyal and Emin~G{\"u}n Sirer.
\newblock Majority is not enough: Bitcoin mining is vulnerable.
\newblock In {\em Financial Cryptography and Data Security}, pages 436--454. Springer, 2014.

\bibitem{FerreiraHWY22}
Matheus V.~X. Ferreira, Ye~Lin~Sally Hahn, S.~Matthew Weinberg, and Catherine Yu.
\newblock Optimal strategic mining against cryptographic self-selection in proof-of-stake.
\newblock In David~M. Pennock, Ilya Segal, and Sven Seuken, editors, {\em {EC} '22: The 23rd {ACM} Conference on Economics and Computation, Boulder, CO, USA, July 11 - 15, 2022}, pages 89--114. {ACM}, 2022.

\bibitem{FerreiraW21}
Matheus V.~X. Ferreira and S.~Matthew Weinberg.
\newblock Proof-of-stake mining games with perfect randomness.
\newblock In P{\'{e}}ter Bir{\'{o}}, Shuchi Chawla, and Federico Echenique, editors, {\em {EC} '21: The 22nd {ACM} Conference on Economics and Computation, Budapest, Hungary, July 18-23, 2021}, pages 433--453. {ACM}, 2021.

\bibitem{FiatKKP19}
Amos Fiat, Anna Karlin, Elias Koutsoupias, and Christos~H. Papadimitriou.
\newblock Energy equilibria in proof-of-work mining.
\newblock In Anna Karlin, Nicole Immorlica, and Ramesh Johari, editors, {\em Proceedings of the 2019 {ACM} Conference on Economics and Computation, {EC} 2019, Phoenix, AZ, USA, June 24-28, 2019}, pages 489--502. {ACM}, 2019.

\bibitem{GarayKL15}
Juan~A. Garay, Aggelos Kiayias, and Nikos Leonardos.
\newblock The bitcoin backbone protocol: Analysis and applications.
\newblock In Elisabeth Oswald and Marc Fischlin, editors, {\em Advances in Cryptology - {EUROCRYPT} 2015 - 34th Annual International Conference on the Theory and Applications of Cryptographic Techniques, Sofia, Bulgaria, April 26-30, 2015, Proceedings, Part {II}}, volume 9057 of {\em Lecture Notes in Computer Science}, pages 281--310. Springer, 2015.

\bibitem{GorenS19}
Guy Goren and Alexander Spiegelman.
\newblock Mind the mining.
\newblock In Anna Karlin, Nicole Immorlica, and Ramesh Johari, editors, {\em Proceedings of the 2019 {ACM} Conference on Economics and Computation, {EC} 2019, Phoenix, AZ, USA, June 24-28, 2019}, pages 475--487. {ACM}, 2019.

\bibitem{HubermanLM21}
Gur Huberman, Jacob Leshno, and Ciamac~C. Moallemi.
\newblock Monopoly without a monopolist: An economic analysis of the bitcoin payment system.
\newblock {\em Review of Economic Studies}, 88:3011--3040, 2021.

\bibitem{KalodnerGCWF18}
Harry Kalodner, Steven Goldfeder, Xiaoqi Chen, S.~Matthew Weinberg, and Edward~W. Felten.
\newblock Arbitrum: Scalable smart contracts.
\newblock In {\em Proceedings of the 27th USENIX Security Symposium (USENIX)}, 2018.

\bibitem{KiayiasKKT16}
Aggelos Kiayias, Elias Koutsoupias, Maria Kyropoulou, and Yiannis Tselekounis.
\newblock Blockchain mining games.
\newblock In {\em Proceedings of the 2016 {ACM} Conference on Economics and Computation, {EC} '16, Maastricht, The Netherlands, July 24-28, 2016}, pages 365--382, 2016.

\bibitem{LaviSZ19}
Ron Lavi, Or~Sattath, and Aviv Zohar.
\newblock Redesigning bitcoin's fee market.
\newblock In {\em The World Wide Web Conference, {WWW} 2019, San Francisco, CA, USA, May 13-17, 2019}, pages 2950--2956, 2019.

\bibitem{LiCT22}
Sheng{-}Nan Li, Carlo Campajola, and Claudio~J. Tessone.
\newblock Twisted by the pools: Detection of selfish anomalies in proof-of-work mining.
\newblock {\em CoRR}, abs/2208.05748, 2022.

\bibitem{LiYT20}
Sheng{-}Nan Li, Zhao Yang, and Claudio~J. Tessone.
\newblock Mining blocks in a row: {A} statistical study of fairness in bitcoin mining.
\newblock In {\em {IEEE} International Conference on Blockchain and Cryptocurrency, {ICBC} 2020, Toronto, ON, Canada, May 2-6, 2020}, pages 1--4. {IEEE}, 2020.

\bibitem{LiYT20b}
Sheng{-}Nan Li, Zhao Yang, and Claudio~J. Tessone.
\newblock Proof-of-work cryptocurrency mining: a statistical approach to fairness.
\newblock In {\em IEEE/CIC International Conference on Communications in China (ICCC Workshops), {ICBC} 2020, Toronto, ON, Canada, May 2-6, 2020}, pages 156--161. {IEEE}, 2020.

\bibitem{Nakamoto08}
Satoshi Nakamoto.
\newblock Bitcoin: A peer-to-peer electronic cash system, 2008.

\bibitem{narayanan2016bitcoin}
Arvind Narayanan, Joseph Bonneau, Edward Felten, Andrew Miller, and Steven Goldfeder.
\newblock {\em Bitcoin and cryptocurrency technologies: a comprehensive introduction}.
\newblock Princeton University Press, 2016.

\bibitem{NeuderMRP19}
Michael Neuder, Daniel~J. Moroz, Rithvik Rao, and David~C. Parkes.
\newblock Selfish behavior in the tezos proof-of-stake protocol.
\newblock {\em CoRR}, abs/1912.02954, 2019.

\bibitem{NeuderMRP20}
Michael Neuder, Daniel~J. Moroz, Rithvik Rao, and David~C. Parkes.
\newblock Defending against malicious reorgs in tezos proof-of-stake.
\newblock In {\em {AFT} '20: 2nd {ACM} Conference on Advances in Financial Technologies, New York, NY, USA, October 21-23, 2020}, pages 46--58. {ACM}, 2020.

\bibitem{Roughgarden21}
Tim Roughgarden.
\newblock Transaction fee mechanism design.
\newblock In P{\'{e}}ter Bir{\'{o}}, Shuchi Chawla, and Federico Echenique, editors, {\em {EC} '21: The 22nd {ACM} Conference on Economics and Computation, Budapest, Hungary, July 18-23, 2021}, page 792. {ACM}, 2021.

\bibitem{SaadNKM19}
Muhammad Saad, Laurent Njilla, Charles Kamhoua, and Aziz Mohaisen.
\newblock Countering selfish mining in blockchains.
\newblock In {\em 2019 International Conference on Computing, Networking and Communications (ICNC)}, pages 360--364. IEEE, 2019.

\bibitem{SapirshteinSZ16}
Ayelet Sapirshtein, Yonatan Sompolinsky, and Aviv Zohar.
\newblock Optimal selfish mining strategies in bitcoin.
\newblock In {\em Financial Cryptography and Data Security - 20th International Conference, {FC} 2016, Christ Church, Barbados, February 22-26, 2016, Revised Selected Papers}, pages 515--532, 2016.

\bibitem{ShiCW23}
Elaine Shi, Hao Chung, and Ke~Wu.
\newblock What can cryptography do for decentralized mechanism design?
\newblock In Yael~Tauman Kalai, editor, {\em 14th Innovations in Theoretical Computer Science Conference, {ITCS} 2023, January 10-13, 2023, MIT, Cambridge, Massachusetts, {USA}}, volume 251 of {\em LIPIcs}, pages 97:1--97:22. Schloss Dagstuhl - Leibniz-Zentrum f{\"{u}}r Informatik, 2023.

\bibitem{WangLLWY21}
Zhaojie Wang, Qingzhe Lv, Zhaobo Lu, Yilei Wang, and Shengjie Yue.
\newblock Forkdec: accurate detection for selfish mining attacks.
\newblock {\em Security and Communication Networks}, 2021:1--8, 2021.

\bibitem{YaishSZ22}
Aviv Yaish, Gilad Stern, and Aviv Zohar.
\newblock Uncle maker: (time)stamping out the competition in ethereum.
\newblock Cryptology ePrint Archive, Paper 2022/1020, 2022.
\newblock \url{https://eprint.iacr.org/2022/1020}.

\bibitem{YaishTZ22}
Aviv Yaish, Saar Tochner, and Aviv Zohar.
\newblock Blockchain stretching {\&} squeezing: Manipulating time for your best interest.
\newblock In David~M. Pennock, Ilya Segal, and Sven Seuken, editors, {\em {EC} '22: The 23rd {ACM} Conference on Economics and Computation, Boulder, CO, USA, July 11 - 15, 2022}, pages 65--88. {ACM}, 2022.

\bibitem{Yao18}
Andrew~Chi{-}Chih Yao.
\newblock An incentive analysis of some bitcoin fee designs.
\newblock {\em CoRR}, abs/1811.02351, 2018.

\end{thebibliography}

\appendix
\section{Selfish Mining in Practice}\label{app:whynot}
\subsection{Relevance of Selfish Mining}
Strategic manipulations, like Selfish Mining, which do not directly affect consensus do not pose the same immediate threat as a double-spend attack. However, they are concerning in practice for (at least) the following two reasons.

\paragraph{Centralizing Force.} Strategic manipulations that provide mild supralinear rewards are a centralizing force: a cohort of miners can achieve greater rewards together than separately. Or, put another way, a sufficiently large miner will achieve greater per-investment rewards than a smaller miner simply due to their size. In general, centralizing forces are undesirable in blockchain applications, as they can lead to centralization over time (which may enable more devastating attacks like double-spends). 

\paragraph{Synergy with consensus-breaking attacks.} Canonical strategic manipulations happen to have synergy with canonical attacks. Consider for example Selfish Mining and double-spending. A Selfish Miner sometimes finds themselves with a hidden chain of several blocks, and launching a double-spend attack at this time is more likely to succeed. Consider also seed manipulation and committee takeover in BFT-based protocols. Manipulating the pseudorandom seed not only helps the deviator lead more rounds, but also helps them control a greater fraction of the BFT committee. If strategic manipulations are profitable, they may additionally lower the cost of stronger consensus-breaking attacks.

\paragraph{Summary.} The above two paragraphs give a brief summary of why strategic manipulations that mildly increase profit of deviators are relevant in practice. Although these manipulations don't directly undermine consensus, they increase the risk the risk of consensus-breaking attacks in the future (via centralizing forces) and in the present (via increased success probability).

\subsection{Factors impacting Selfish Mining}

It is worth distinguishing deviations from consensus protocols in the the following manner: does the deviator's profit come at the cost of other participants in consensus, or at the cost of a user of the protocol? For example, Selfish Mining steals rewards from other miners, whereas double-spending steals cryptocurrency from a user. Numerous attacks at the cost of users have been detected across multiple blockchains. To the best of our knowledge,~\cite{YaishSZ22} recently discovered the first instance of a strategic manipulation of a consensus protocol whose profits were primarily at the expense of other consensus participants. In particular, to the best of our knowledge, no evidence of Selfish Mining in practice has been discovered.\footnote{It is not entirely clear how much effort has been spent aiming to detect Selfish Mining in practice, but we are not aware of any successful detections.} 

Note that this holds on Bitcoin, despite periods of time (including at the time of writing) when a single mining pool controls $>1/3$ of Bitcoin's total hashrate. This also holds on much smaller ``altcoins'' that are vulnerable to the same attack, despite the low cost of a single entity acquiring $>1/3$ of the total hashrate. The following is a summary of several potential explanations for this put forth by the community:

\begin{enumerate}[topsep=0pt]
\item It is possible that no single entity has had sufficient hashpower/network connectivity for Selfish Mining to be profitable on large cryptocurrencies like Bitcoin.\footnote{Note that a public mining pool is not a ``single entity'' --- see later bullet.} 

It is worth observing some limitations of this argument, beyond the present state of large cryptocurrencies. For example, this argument does not apply well to altcoins, as it is significantly less expensive to control a $>1/3$ fraction of the hashrate. Moreover, it explains why Selfish Mining \textit{has perhaps not yet occurred on Bitcoin}, but does not explain why it \textit{forever will not occur on Bitcoin}. Indeed, recall that with sufficiently high network connectivity, the hashrate required for Selfish Mining to be profitable can drop all the way to $0$ (so confirming that Selfish Mining is not a threat would require constant estimations of hashrate and network connectivity of large miners, which may be problematic in a pseudonymous protocol). 
\item There are significant engineering challenges to running a mining operation. Moreover, implementing Selfish Mining requires interacting with ASICs, and may not be as simple as its pseudocode. Therefore, the small boost to profit may not outweigh the opportunity cost of tackling more significant engineering challenges.\footnote{The authors learned of this possibility from Arvind Narayanan and Aviv Zohar, and thank them for suggesting it.} 

It is also worth observing some limitations of this argument, beyond the present state of cryptocurrency mining. Indeed, the cryptocurrency mining industry is still rapidly evolving, and the opportunity cost of ignoring other engineering challenges is high. But as the industry stabilizes, miners might have fewer options to improve profit margins.

\item If Selfish Mining were traceable to the deviator, the deviator might suffer targeted repercussions. In particular, note that Selfish Mining by a public mining pool would be traceable to the operator, because the mining pool has to ask its participants to Selfish Mine on its behalf\footnote{Observe that in order to successfully Selfish Mine, the Selfish Miner \emph{must not immediately broadcast blocks after creation}. Therefore, if a public mining pool wishes to Selfish Mine, they must request that members of the pool not broadcast their blocks immediately upon creation. Moreover, in order to successfully Selfish Mine, the Selfish Miner must sometimes \emph{point to a block that has not yet been broadcast} when creating a new block. Therefore, if a public mining pool wishes to Selfish Mine, they must request that members of the pool not broadcast the existence of the block to which their own created blocks will point. Due to both of these, any members of a public mining pool will detect that the mining pool is attempting to Selfish Mine (and because the mining pool is public, it will be hard to keep this hidden).} Therefore, miners may leave the pool, eliminating its ability to Selfish Mine in the first place. Similarly, the community (or perhaps a legal enforcement agency) could decide to censor all coins owned by the deviator. Note that targeted repercussions are certainly not guaranteed,\footnote{Indeed, we've previously noted that~\cite{YaishSZ22} recently identify a deviant mining pool on Ethereum. F2Pool quickly admitted to deviation, but it remains unclear if F2Pool has suffered as a result.} but the possibility may be sufficient to deter a risk-averse miner.

The limitation of this argument is simply that Selfish Mining is untraceable in absence of outside information (such as instructions from a public pool). So, while this argument applies well to public mining pools, it does not apply well to large individual miners (who can make a new public key for every block they mine).
\item Selfish Mining is statistically detectable, and the value of the underlying cryptocurrency may be (significantly) negatively affected as a result. Therefore, while Selfish Mining may be profitable when denominated in the underlying cryptocurrency, it could be wildly unprofitable when denominated in an objective unit of value (e.g. USD). 

This argument applies broadly to any miner considering any statistically detectable deviation. However, it still has limitations. For example, it relies on the miner's belief that detection will negatively impact the underlying cryptocurrency's value and inability to hedge against this,\footnote{It would be interesting for future work to estimate whether the value of ETH was impacted by~\cite{YaishSZ22}'s discovery and F2Pool's admission, although currently this remains unclear.} or the miner's risk-aversion. After our work, another limitation is that perhaps the deviation is statistically undetectable (in which case the market cannot possibly react to its detection).
\end{enumerate}

\noindent In summary, the above bullets provide answers to the following questions:

\noindent Despite theoretical profitability in a stylized model, \textbf{why might it be the case that no individual entity has Selfish Mined on Bitcoin?} Bullet One argues the possibility that no single entity ever possessed the hashrate/network connectivity to profit. Bullet Two argues that even if they did, the modest gains may not justify the engineering opportunity cost. Bullet Four further argues that even if the math works out when denominated in bitcoin, there is significant uncertainty and high risk of significant losses when denominated in USD.

\noindent Despite theoretical profitability in a stylized model, \textbf{why might it be the case that no individual entity has Selfish Mined on smaller altcoins?} Bullet Two argues that, even with sufficient resources to profitably Selfish Mine, the modest gains may not justify the engineering opportunity cost. Bullet Four further argues that even if the math works out when denominated in altcoin, there is significant uncertainty and high risk of significant losses when denominated in USD.

\noindent Despite theoretical profitability in a stylized model, \textbf{why might it be the case that no public mining pool has Selfish Mined on any cryptocurrency?} Bullet Two argues that, even with sufficient resources to profitably Selfish Mine, the modest gains may not justify the engineering opportunity cost. Bullet Three argues that there is significant uncertainty and high risk of significant retaliation (denominated in either the underlying cryptocurrency, or USD). Bullet Four further argues that even if the math works out when denominated in bitcoin, there is significant uncertainty and high risk of significant losses when denominated in USD.

They key question of interest to mechanism design researchers with an eye towards fundamental concepts not overfit to the current state of affairs is the following: given that Selfish Mining is theoretically profitable in a stylized model but has not been detected in practice, \textbf{is there a realistic concern of Selfish Mining in the future?} Prior to our work, Bullet Four is (in the authors' opinion, and anecdotally from community discussions) perhaps the strongest argument in the negative direction. Our work pushes this answer closer to the affirmative, and serves as practical motivation for the design and analysis of updated consensus protocols with stronger incentive guarantees.

Finally, note that while we have phrased this key question in the language of Selfish Mining in longest-chain proof-of-work cryptocurrencies, and our theorems hold only in this setting, the same principles apply broadly across the blockchain and web3 ecosystem. Indeed many other consensus protocols admit profitable deviations~\cite{CarlstenKWN16, KiayiasKKT16, SapirshteinSZ16, KalodnerGCWF18, BrownCohenNPW19,  FerreiraW21, FerreiraHWY22, YaishSZ22, YaishTZ22}, and many other aspects of the broader ecosystem currently leave (small numbers of) individual entities with the power to take significant harmful actions. However, it is commonly argued that participants in consensus protocols and powerful individuals elsewhere in the ecosystem are disincentivized from deviant behavior by identical reasoning to Bullet Four: it may be profitable when denominated in the underlying cryptocurrency/token/etc., but carry an uncertain risk of massive unprofitability when denominated in USD. Our work suggests that further care should be taken in such settings to understand whether such deviations are indeed detectable.
\section{Brief Comparison to Richer Latency Models}\label{app:latency}

Here, we first briefly note that all prior work we are aware of concerning strategic block withholding consider stylized latency models (essentially, our $0$-NCG)~\cite{EyalS14, KiayiasKKT16,SapirshteinSZ16, CarlstenKWN16, BrownCohenNPW19, NeuderMRP19, NeuderMRP20,FerreiraW21, FerreiraHWY22}. Still, it is worth a brief effort to connect this common stylized model to richer models from distributed computing, such as~\cite{GarayKL15, BirmpasKLC20}.

In comparison to~\cite{GarayKL15}, which is an adversarial latency model, our Attacker could certainly be implemented by their adversary, and therefore all of our results would hold in their model. But, our Attacker is using nowhere near their adversary’s capability (in particular, only the ``rushing'' aspect of their adversary is relevant to our Attacker). For example, to ``win all ties'' in our model, after finding a block their adversary could wait until an honest party finds a block, and then reorder all Receive() strings to place their own block first. The adversary does not need to delay messages outside of this reordering, nor corrupt the content of any messages. To ``lose all ties'' in our model, their adversary does not need to take any malicious actions, and just needs to be aware of all parties’ Receive() strings to time their own broadcasts.

In comparison to~\cite{BirmpasKLC20}, which is a stochastic latency model, our model is equivalent to their special case with $q_i = 1$ for all $i$. The two key differences to the most general version of the their model are: (a) their model is more complex in that it allows for different miners to experience different latency, and (b) in our model, every player is always within one block of being up to date, whereas players in the~\cite{BirmpasKLC20} model are sometimes multiple blocks behind. 
\section{Exact Analysis of Rewards Using Markov Chains}\label{app:exactreward}

To do the award analysis using a Markov Chain, we'll simplify analysis by counting \emph{only the Paired rounds, and who wins}. See Lemma~\ref{lem:pair} for why this suffices.

To analyze the strategy, we'll need a state to keep track of the following information:
\begin{itemize}[nosep]
\item How many blocks does the attacker have hidden that we \emph{know} will form a Pair?
\item Is the last hidden block of the attacker Pair, Single, or Undecided? (P, S, U). If there are no hidden blocks, we leave this blank.
\end{itemize}

We will transition \emph{every time the status of a new height is determined}. Note that several blocks may be created during a single transition. We do this for the following reason. Recall that our undetectability analysis computed a probability $P_h$ at the moment that $S(h-1)$ was determined. It is therefore simplest if we only have transitions from the moment $S(h-1)$ is determined to the moment $S(h)$ is determined (and nothing in between). This will let us a) compute $P_h$ the moment $S(h-1)$ is determined, and b) use $P_h$ explicitly in the transition to where $S(h)$ is determined.

Table~\ref{tab:transitions} summarises all transitions the Markov chain might take. The column $q_e$ refers to the probability of taking this transition, which fully defines the Markov chain. We introduce three counting schemes, which will let us confirm that calculations are done correct. $H^p_e$ and $S^p_e$ count the number of pairs we learn are guaranteed to be won by the honest party and attacker, respectively, during this transition. $H^b_e$ and $S^b_e$ count the number of blocks in the longest chain (part of a pair or not) that we learn are guaranteed to be won by the honest party and attacker, respectively, during this transition. $H^s_e$ and $S^s_e$ count the number of solo pairs we learn are guaranteed to be won by the honest party and attacker, respectively, during this transition.

\newcolumntype{C}{>{\arraybackslash}p{.95\textwidth}}

\begin{table}%
    \caption{The Markov Chain describing the attacker's strategy in the SP-Simple model, along with brief explanations of the events corresponding to each transition. The transitions out of State 0 are ommitted, as they are identical to transitions out of state $i$S with $i=0$. The respective values of $P_h$ are given in Table~\ref{tab:Ph}.}
    \label{tab:transitions}
        \begin{center}
        \begin{tabularx}{\textwidth}{ XXcccccc }
            \toprule
            Transition & $q_e$ & $H^p_e$ & $S^p_e$ &$H^b_e$&$S^b_e$&$H^s_e$&$S^s_e$ \\ \midrule
            ($i$S,0) & $(1-\alpha)^{i+1}$ & 0 & 0 & 0 & 1 & 0 & 0  \\
            \multicolumn{8}{C}{Honest completes all pairs plus an extra block, then Attacker finds a block. Attacker publishes all hidden blocks.}\\ \hline
            ($i$S,$(i-j+1)$P) & $(1-\alpha)^j\alpha \frac{\beta}{P_h}$ & 0 & 0 & 0 & 0 & 0 & 0 \\
            \multicolumn{8}{C}{Honest completes $j$ pairs, then Attacker finds a block. The  Attacker immediately decides the new block is Pair.}\\ \hline
            ($i$S,$(i-j)$S) & $(1-\alpha)^j\alpha (1-\frac{\beta}{P_h})$ & 0 & 0 & 0 & 1 & 0 & 0 \\ 
            \multicolumn{8}{C}{Honest completes $j$ pairs, then Attacker finds a block. The  Attacker immediately decides the new block is Single.}\\ \hline
              ($i$U, 0) & $(1-\alpha)^i$ & 0 & 0 & 0 & 0 &0 &0 \\
            \multicolumn{8}{C}{Honest finds the next $i$ blocks. The last block is pivotal, and Attacker publishes it as Single.}\\ \hline
            ($i$U,$(i+1-j)$U) & $(1-\alpha)^j\alpha\frac{\beta}{P_h}$ & 0 & 1 & 0 & 1 & 0 & 0 \\
            \multicolumn{8}{C}{Honest finds $j$ blocks and then Attacker finds a block. Attacker decides (in advance) that the undecided block is Pair. The new Attacker block might be pivotal.}\\ \hline
            ($i$U,$(i+1-j)$P) & $(1-\alpha)^j\alpha(1-\frac{\beta}{P_h})\beta$ & 0 & 0 & 0 & 0 & 0 & 0 \\ 
            \multicolumn{8}{C}{Honest finds $j$ blocks and then Attacker finds a block. Attacker decides that the undecided block is Single and the new block is Pair.}\\ \hline
            ($i$U,$(i+1-j)$S) & $(1-\alpha)^j\alpha(1-\frac{\beta}{P_h})(1-\beta)$ & 0 & 0 & 0 & 1 & 0 & 0 \\
            \multicolumn{8}{C}{Honest finds $j$ blocks and then Attacker finds a block. Attacker decides that the undecided block is Single and the new block is Single.}\\ \hline
            ($i$P,0) & $(1-\alpha)^{i+1}$ & $1-\gamma$ & $\gamma$ & $2-\gamma$ & $\gamma$ & $1-\gamma$ & $\gamma$ \\
            \multicolumn{8}{C}{Honest completes all pairs, triggers a race, which Honest wins. Tiebreaks for Attacker w.p. $\gamma$}\\ \hline
            ($i$P,0) & $(i+1)(1-\alpha)^i\alpha$ & 0 & 1 & 0 & 2 & 0 & 1 \\
            \multicolumn{8}{C}{Honest completes all pairs, triggers a race, which Attack wins.}\\ \hline
            ($i$P, $(i+1-j)$U) & $(j+1)\alpha^2(1-\alpha)^{j}\frac{\beta}{P_h}$ & 0 & 2 & 0 & 3 & 0 & 0 \\ 
            \multicolumn{8}{C}{Honest finds $j$ blocks before Attacker finds 2 blocks $a,b$, and the first Attacker block $a$ is marked as Pair. Guarantees win of last hidden pair and $a$ for Attacker. $b$ will be in longest chain.}\\ \hline
            ($i$P, $(i+1-j)$P) & $(j+1)\alpha^2(1-\alpha)^{j}(1-\frac{\beta}{P_h})\beta$ & 0 & 1 & 0 & 2 & 0 & 1 \\
            \multicolumn{8}{C}{Honest finds $j$ blocks before Attacker finds 2 blocks $a,b$, with $a$ is marked as Single and $b$ as Pair. Guarantees win of last hidden pair by Attacker, but the fate of $b$ is unknown.}\\ \hline
            ($i$P, $(i-j)$S) & $(j+1)\alpha^2(1-\alpha)^{j}(1-\frac{\beta}{P_h})(1-\beta)$ & 0 & 1 & 0 & 3 & 0 & 1 \\
            \multicolumn{8}{C}{Honest finds $j$ blocks before Attacker finds 2 blocks $a,b$ and marks both as Single. Guarantees win of last hidden pair by Attacker.}\\
            \bottomrule
            \end{tabularx}
        \end{center}
             \bigskip    \bigskip
\end{table}%

\begin{table}
    \caption{The values of $P_h$ in the attacker's strategy in the SP-Simple Model, described in Table~\ref{tab:transitions}. Justification for these expressions is given in Lemma~\ref{lem:Ph-justify}}\bigskip
    \label{tab:Ph}
    \begin{minipage}{\columnwidth}
        \begin{center}
            \begin{tabular}{ cl }
                \toprule
                Transition & $P_h$ \\ \midrule
                ($i$S,$\cdot$) & $1-(1-\alpha)^{i+1}$ \\
                ($i$U,$\cdot$) & $1-(1-\alpha)^i$ \\
                ($i$P,$\cdot$) & $1-(1-\alpha)^{i+1}-(i+1)\alpha(1-\alpha)^i$ \\
                \bottomrule
            \end{tabular}
        \end{center}
    \end{minipage}
\end{table}

\begin{lemma}\label{lem:Ph-justify} Table~\ref{tab:Ph} gives the correct expressions for computing $P_h$ in SP-Simple game.
\end{lemma}
\begin{proof}
    Recall that in the SP-Simple model, $P_h$ is defined as the probability that Attacker creates a block of height $h$ and that block is safe, conditioned on all information available \emph{as of the first moment we know} $S(h')$ for all $h' < h$. Recall also that transitions in the Markov chain happen from when $S(h-1)$ is determined to when $S(h)$ is determined. When the originating state in a transition is $i$S or $i$P, $h-1$ is the height of the last hidden block, since its state is known to be S or P, and $h$ is the new height whose state is determined via the transition. When the originating state in a transition is $i$U, height $h$ refers to \emph{the undecided block} whose state is determined via the transition.
\begin{itemize}
    \item{($i$S,$\cdot$) or $(0,\cdot)$}: Height $h-1$ is Single, so if Attacker finds the first block of height $h$, it will be safe. Attacker will mine the first block of height $h$ unless Honest completes all $i$ Pairs, and then finds an extra block, all before Attacker finds a block. Therefore, $P_h=1-(1-\alpha)^{i+1}$.

    \item{($i$U,$\cdot$)}: The Undecided block at height $h$ will be safe unless Honest completes all $i$ Pairs before Attacker finds a block, making the undecided block pivotal. Therefore, $P_h=1-(1-\alpha)^{i}$.

    \item{($i$P,$\cdot$)}: The event measured by $P_h$ holds \emph{unless} one of the following (mutually exclusive) events occures: 
    \begin{itemize}
        \item The block at height $h$ is mined by Honest. This happens if Honest completes the $i$ Pairs and finds another block before Attacker finds any blocks, which happens with probability $(1-\alpha)^{i+1}$.
        \item The block at height $h$ is mined by Attacker but is Pivotal; that is, Honest finds a block of height $h-1$ before attacker finds a block of height $h+1$. This happens iff among the next $i+1$ blocks to be discovered, Honest finds $i$ and Attacker finds one. The probability of this is $(i+1)\alpha(1-\alpha)^i$.
    \end{itemize}
    Combining the two implies $P_h1-(1-\alpha)^{i+1}-(i+1)\alpha(1-\alpha)^i$.\qedhere
\end{itemize}
\end{proof}

\begin{theorem}[Markov Chain Specification]
The following system of equations characterizes the stationary distribution $\mathbf{p}$ of the Markov chain in Table~\ref{tab:transitions}:
\begin{align*}
    p_{0} &=
        \sum_{j=0}^\infty (1-\alpha)^{j+1} \cdot p_{j,S}
        + \sum_{j=1}^\infty \left((1-\alpha)^{j+1}+(j+1)\alpha(1-\alpha)^j\right)\cdot p_{j,P}
        + \sum_{j=2}^\infty(1-\alpha)^j \cdot p_{j,U}\\
    p_{i,S} &=
        \sum_{j=0}^\infty (1-\alpha)^{j}\alpha\left(1-\frac{\beta}{1-(1-\alpha)^{i+j+1}}\right)\cdot p_{i+j,S}\\
        &+ (j+1)\alpha^2(1-\alpha)^{j}\left(1-\frac{\beta}{1-(1-\alpha)^{i+j+1}-(i+j+1)\alpha(1-\alpha)^{i+j}}\right)(1-\beta) \cdot p_{i+j,P}\\
        &+ \alpha(1-\alpha)^j \left(1-\frac{\beta}{1-(1-\alpha)^{i+j}}\right)(1-\beta) \cdot p_{i+j,U} \tag{\text{for $i\geq 1$}}\\
    p_{i,U} &=
        \sum_{j=0}^\infty (j+1)\alpha^2(1-\alpha)^{j} \frac{\beta}{1-(1-\alpha)^{i+j+1}-(i+j+1)\alpha(1-\alpha)^{i+j+1}}(1-\beta)\cdot p_{i+j+1,P}\\
        &+ (1-\alpha)^j\alpha \cdot \frac{\beta}{1-(1-\alpha)^{i+j+1}} \cdot p_{i+j+1,U}\tag{\text{for $i\geq 2$}}\\ 
    p_{i,P} &=
        \sum_{j=0}^\infty (1-\alpha)^{j}\alpha\cdot\frac{\beta}{1-(1-\alpha)^{i+j-1}}\cdot p_{i+j-1,S}\\
        &+ (j+1)\alpha^2(1-\alpha)^{j}\left(1-\frac{\beta}{1-(1-\alpha)^{i+j+1}-(i+j+1)\alpha(1-\alpha)^{i+j+1}}\right)\beta\cdot p_{i+j+1,P}\\
        &+ (1-\alpha)^j\alpha \left(1-\frac{\beta}{1-(1-\alpha)^{i+j+1}}\right)\beta\cdot p_{i+j+1,U}\tag{\text{for $i\geq 1$}}
\end{align*}
\end{theorem}

\begin{theorem}[Reward Specification]
Let $\{p_u\}_u$ denote the stationary distribution of the Markov Chain described in Table~\ref{tab:transitions}. For any transition $e$ from state $u$ to state $v$, let $p_e:=p_u\cdot q_e$ denote the fraction of transitions spent taking transition $e$. Then the following three equations each compute the expected reward of Undetectable Selfish Mining:

$$\frac{\sum_{e} p_e \cdot S^b_e}{\sum_e p_e \cdot (H^b_e+S^b_e)},$$
$$\frac{\alpha - \beta + \beta \cdot \frac{\sum_{e} p_e \cdot S^p_e}{\sum_e p_e \cdot (H^p_e+S^p_e)}}{1-\beta}.$$
$$\frac{\alpha - \beta + \beta \cdot \left(\frac{\sum_{e} p_e \cdot S^s_e}{\sum_e p_e \cdot (H^b_e+S^s_e)}+\left(2-\frac{\sum_{e} p_e \cdot S^s_e}{\sum_e p_e \cdot (H^b_e+S^s_e)}\right)\beta -\left(1-\frac{\sum_{e} p_e \cdot S^s_e}{\sum_e p_e \cdot (H^b_e+S^s_e)}\right)\beta^2\right)}{1-\beta}$$
\end{theorem}

\begin{proof}
The first equation follows by the exact same logic as classical Selfish Mining analysis. The second follows by Lemma~\ref{lem:pair}. 

The final equation follows by the following calculations, building off Lemma~\ref{lem:pair}. Observe that Undetectable Selfish Mining has the following properties:
\begin{itemize}[nosep]
\item If there are ever multiple Pair rounds in a row, USM wins them all.
\item If there is a solo Pair round, USM may win or lose it.
\item The number of Pair rounds in a row (between two Single rounds) is distributed independently across time, and equal to $i$ with probability $(1-\beta)\cdot \beta^i$. 
\end{itemize}
\noindent
Therefore, we can write the expected number of Pair rounds between two Single rounds as:
$$\sum_{i=0}^\infty (1-\beta)\beta^i \cdot i = \frac{\beta}{1-\beta}.$$
\noindent
And the expected number of Pair rounds won by the attacker between two Single rounds as:
$$w\cdot (1-\beta)\beta + \sum_{i=2}^\infty (1-\beta)\beta^i \cdot i = \frac{(2-\beta)\beta^2}{1-\beta}+w\beta(1-\beta) = \frac{2\beta^2 - \beta^3 + w\beta - 2w\beta^2+w\beta^3}{1-\beta}.$$
\noindent
Therefore, the expected fraction of Pairs won by the attacker is:
\begin{align*}
\frac{\frac{2\beta^2 - \beta^3 + w\beta - 2w\beta^2+w\beta^3}{1-\beta}}{\frac{\beta}{1-\beta}} = w+(2-w)\beta -(1-w)\beta^2.
\end{align*}
\end{proof}

\section{A Simpler Selfish Mining Analysis}\label{sec:simplerSM}
    For readers who are mostly interested in when Selfish Mining is strictly profitable (but not by how much), we provide a shorter analysis of the classical selfish mining strategy that avoids analyzing a Markov chain.

    \begin{lemma}[\cite{EyalS14}] When using the Selfish Mining Strategy that wins a $\gamma$ fraction of ties, Attacker wins more than a $(1-\alpha)$ fraction of Pair rounds if and only if $\alpha \geq \frac{1-\gamma}{3-2\gamma}$.
    \end{lemma}
    \begin{proof}
    Let's count the ways that a Paired round can arise. Starting from state $0$, the attacker must find the next block (otherwise, we produce a Single honest round). Conditioned on this, there are three possible cases:
    \begin{itemize}
    \item The next two blocks are honest. This happens with probability $(1-\alpha)^2$, and results in a winning pair with probability $\gamma$.
    \item The next block is honest, and the second block is attacker. This happens with probability $(1-\alpha)\alpha$, and results in a winning pair with probability $1$.
    \item The next block is from the attacker. From here, there will be at least one Pair, and the attacker will win all pairs. This occurs with probability $\alpha$. The expected number of Pairs, conditioned on this, is $1+\frac{\alpha}{1-2\alpha}$.\footnote{To quickly see this, observe that the number of Pairs is equal to one plus the number of blocks the attacker finds before it loses its lead of one. This latter number $x$ satisfies the recurrence $x = (1-\alpha)\cdot 0 + \alpha \cdot (2x+1)$: if the honest finds a block immediately, $x=0$. If the attacker finds the next block, we now need to lose a lead of one twice.}
    \end{itemize}

    Thus, every time the attacker finds a block after a reset, the expected number of pairs created is:
    \[(1-\alpha)^2\cdot 1 + \alpha(1-\alpha) \cdot 1 + \alpha\cdot \frac{1-\alpha}{1-2\alpha}.\]

    And the expected number of winning pairs created is:

    $$(1-\alpha)^2\cdot \gamma + \alpha(1-\alpha) \cdot 1 + \alpha\cdot \frac{1-\alpha}{1-2\alpha}.$$

    Observe that when $\gamma = 1$, their ratio (unsurprisingly) is $1 > 1-\alpha$ for all $\alpha > 0$. Observe also that when $\gamma = 0$, substituting in $\alpha = 1/3$ results in a ratio of $2/3 = 1-\alpha$. For general $\gamma$, observe that their ratio exceeds $(1-\alpha)$ if and only if:

    \begin{align*}
    \frac{(1-\alpha)^2\cdot \gamma + \alpha(1-\alpha) \cdot 1 + \alpha\cdot \frac{1-\alpha}{1-2\alpha}}{(1-\alpha)^2\cdot 1 + \alpha(1-\alpha) \cdot 1 + \alpha\cdot \frac{1-\alpha}{1-2\alpha}} &\geq 1-\alpha\\
    \Leftrightarrow \frac{(1-2\alpha)(1-\alpha)^2\cdot \gamma + \alpha(1-\alpha)(1-2\alpha)+\alpha(1-\alpha)}{(1-2\alpha)(1-\alpha)^2\cdot 1 + \alpha(1-\alpha)(1-2\alpha)+\alpha(1-\alpha)} &\geq 1-\alpha\\
    \Leftrightarrow \frac{(1-2\alpha)(1-\alpha)\cdot \gamma + \alpha(1-2\alpha)+\alpha}{(1-2\alpha)(1-\alpha)\cdot 1 + \alpha(1-2\alpha)+\alpha} &\geq 1-\alpha\\
    \Leftrightarrow \frac{(1-2\alpha)(1-\alpha)\cdot \gamma + 2\alpha(1-\alpha)}{(1-2\alpha)(1-\alpha)\cdot 1 + 2\alpha(1-\alpha)} &\geq 1-\alpha\\
    \Leftrightarrow \frac{(1-2\alpha)\cdot \gamma + 2\alpha}{(1-2\alpha)\cdot 1 + 2\alpha} &\geq 1-\alpha\\
    \Leftrightarrow (1-2\alpha)\cdot \gamma + 2\alpha &\geq 1-\alpha\\
    \Leftrightarrow (3-2\gamma)\alpha &\geq 1-\gamma\\
    \Leftrightarrow \alpha &\geq \frac{1-\gamma}{3-2\gamma}.
    \end{align*}
    \end{proof}

\section{Generalizing the SP-Simple Model}\label{app:general}
For simplicity of notation, we'll let $\beta'$ refer to the probability that both players create a block, $\alpha'$
refer to the probability that only the strategic player creates a block. Note that
$1-\alpha' - \beta'$
is the probability that only the honest player creates a block. Solving $\alpha=\frac{\beta'+\alpha'}{1+\beta'}$ for $\alpha'$, we get $\alpha'=\alpha(1+\beta')-\beta'$.

\subsection{SP-Simple vs. SP-General}
In the SP-Simple model, the attacker scales the bias of the their coin flip by $P_h$ to account for the possibility of their block becoming pivotal and thus a \emph{Forced Single}. The scaling serves to up-weight the probability of a Pair label in cases where the block does in fact end up safe.

In the SP-General model, there is still a possibility of Forced Singles when an attacker block ends up pivotal. Additionally, there is also the possibility of \emph{Forced Pairs} since forks are possible even if the attacker is honest. This imposes an additional constraint on ranges of $\beta$ with respect to $\beta'$ (Lemma~\ref{lem:validity-general}).

\subsection{The Strategy}
We'll again break down the strategy into a labeling strategy first, and a broadcasting strategy second.

\begin{definition}[$P$s and $Q$s]\label{def:P-Q}
For each height $h\geq 1$, define $P_h$ to be the probability that the attacker creates a safe block at height $h$, and $Q_h$ to be the probability that there is a Forced Pair at height $h$.
\end{definition}

\begin{definition}[Conditional $P$s and $Q$s]\label{def:con-P-Q}
For each height $h\geq 1$ and state sequence \footnote{For simplicity, we will omit the label of the Genesis block, and start all state sequences at height 1 from here on.} $s\in\{\text{Single, Pair}\}^*$, we write $P_h\big\rvert_s$ to denote the probability that the attacker creates a safe block at height $h$, conditioned on the state sequence $s$. Similarly, we write $Q_h\big\rvert_s$ to denote the probability that there is a Forced Pair at height $h$, conditioned on the state sequence $s$.
\end{definition}

\begin{definition} Let $s(h)$ be the sequence of states $(S(1),\ldots,S(h))$.
\end{definition}

\begin{definition}[Generalized Labeling Strategy]\label{def:labelling-general}
Given $P$ and $Q$ defined in Definition~\ref{def:P-Q}, the generalized labeling strategy is then as follows:
\begin{itemize}
    \item If honest creates the first block at height $h$, then $h$ is labeled Single.
    \item If both players find their first block at height $h$ simultaneously, then $h$ is labeled Pair. We call this a \emph{Forced Pair}.
    \item If the attacker creates the first block at height $h$, and it is pivotal, then $h$ is labeled Single.
    \item If the attacker creates the first block at height $h$, and it is safe, then $h$ is labeled Pair with probability $(\beta-Q_h\big\rvert _{s(h-1)})/{P_h\big\rvert_{s(h-1)}}$, and Single otherwise. Here $s$ denotes the sequence of states of all blocks of height strictly smaller than $h$.
\end{itemize}
Again, the instant we have the necessary information to label a height, we do so.
\end{definition}

The labeling strategy here is pretty similar to the labeling strategy in the SP model. The second bullet is new in the Generalized SP model and corresponds to the creation of Forced Pairs. The last bullet is a modification of its counterpart in the SP Model, except that the probability of labeling a safe block as Pair is $(\beta-Q_h\big\rvert _{s(h-1)})/{P_h\big\rvert_{s(h-1)}}$ instead of $\beta/P_h$. The lower probability balances the higher incidence of forks due to Forced Pairs in the Generalized SP model. The main difference between the two models is that unlike in the SP labeling strategy, here we do not condition on \emph{all the available information}, rather \emph{only on the states of all blocks at a smaller height}.

Before diving into the analysis, we will outline a list of conceptually distinct equivalence classes (called ``worlds''), each representing a possible path that got us here along with all the relevant information to compute relevant quantities. At a high level, these worlds allow us to express additional conditioning on information beyond just the states of blocks of lower height. While the labelling strategy itself is agnostic about these worlds (since all its randomness is a function of states of lower blocks only), its analysis requires treating these worlds separately.

\begin{definition}
    For each $h$, let $t(h)$ be the round in which the first block of height $h$ is mined.
\end{definition}
\newpage
\begin{definition}[Worlds with respect to $h$]\label{def:worlds}
For any $h$, we can partition possible realizations of the randomness in block production into the following set of exhaustive and disjoint \emph{worlds with respect to $h$}:
\begin{itemize}[label=,leftmargin=*,nosep]
    \item $\mathbf{(A_h)}$ Only Honest creates a block at time $t(h)$.
    \item $\mathbf{(B_h)}$ Both players simultaneously create blocks at height $h$ during time $t(h)$.
    \item  \quad \ \ \ \ Only Attacker creates a block at time $t(h)$, and
    \begin{itemize}[label=,leftmargin=10mm]
        \item \quad \ \ \ \ $S(h-1)$ is Single, and
        \begin{itemize}[label=,leftmargin=15mm]
            \item $\mathbf{(C_h)}$ $S(h)$ is Single.
            \item $\mathbf{(D_h)}$ $S(h)$ is Pair.
        \end{itemize}
        \item $\mathbf{(E_h)}$ $S(h-1)$ is Pair, and Honest has published a block of height $h-1$ by $t(h)$ (including possibly $t(h)$, because Honest and Attacker both found blocks during $t(h)$ \footnote{Recall that honest does \emph{not} yet have a block of height $h$, as this case deals only with the case that attacker finds the first block of height $h$ alone --- we have already handled Forced Pairs in an earlier case.}
        \item \quad \ \ \ \ $S(h-1)$ is Pair, and honest has \emph{not} published a block of height $h-1$ by time $t(h)$, and
        \begin{itemize}[label=,leftmargin=15mm]
            \item $\mathbf{(F_h)}$ Honest finds a block at height $h-1$ before Attacker mines at height $h+1$.
            \item $\mathbf{(G_h)}$ Attacker finds a block at height $h+1$ before Honest mines at height $h-1$.
            \item $\mathbf{(H_h)}$ both events happen at the same time (\textit{i.e.} during time $t(h+1)$, Honest finds a blog of height $h-1$ and Attacker finds a block of height $h+1$).
        \end{itemize}
    \end{itemize}
\end{itemize}
We call each of $\mathbf{A_h, B_h,\ldots, H_h}$ a \emph{world}, and denote by $W_h$ the set of all (eight) worlds with respect to height $h$.
\end{definition}

Given the worlds outlined in Definition~\ref{def:worlds}, we can define $P^w$ and $Q^w$ as the counterparts to $P$ and $Q$ but \emph{conditioned on world $w$} (in addition to the normal conditioning on previous states). This is formalized in the next definition.

\begin{definition}[Conditional $P$s and $Q$s]\label{def:conditional-P-Q}
    For any height $h\geq 1$, any world $w\in W_h$, and any state sequence $s\in\{\text{Single},\text{Pair}\}^*$, define $P_h\big\rvert_{w,s}$ as the probability that the attacker creates a safe block at height $h$, conditioned on \emph{both} the states matching $s$ \emph{and} the world with respect to $h$ being $w$ (and 0 if the event being conditioned on is empty). Similarly, define $Q_h\big\rvert_{w,s}$ as the probability that there is a forced pair at height $h$, conditioned on \emph{both} the states matching $s$ \emph{and} the world $w$ (and 0 if the event being conditioned on is empty).
\end{definition}

\begin{claim}[When labels are determined]\label{claim:when-label-general}
    For all $h$, there is sufficient information to determine $S(h)$ by the moment that a block of height $h+1$ is created.
\end{claim}
\begin{proof}
    The proof is very similar to proof of Claim~\ref{claim:when-label}, with two added cases to handle Forced Pairs (bullets (2) and (3)(c) below). We will include the full proof of all cases for completeness.
    
    We'll proceed by induction on $h$. The base case holds vacuously when $h=0$, since the genesis block is Single by definition.

    Let's first consider the moment $t(h)$ when the first block of height $h$ is created. We break down the analysis into the following exhaustive set of cases:
    \begin{enumerate}
        \item The first block of height $h$ is created by honest. Then $S(h)$ is certainly Single (and we learn this immediately when the first block of height $h$ is created, which is by the time the first block of height $h+1$ is created).
        \item The first blocks of height $h$ are created simultaneously by both attacker and honest. Then $h$ is a Forced Pair, and we immediately learn that $S(h)$ is Pair (again, the moment the blocks of height $h$ are created, which is before the first block of height $h+1$ is created).
        \item The attacker finds the first block of height $h$ (and honest does not yet have a block of height $h$ at this moment). By the inductive hypothesis, we certainly know the labels of all blocks of height $h' < h$. This in particular implies that we have all the necessary information to compute $P_h\big\rvert_s(h-1)$ and $Q_h\big\rvert_s(h-1)$, and can already flip any coins that might be needed to label $h$, in case $h$ is safe. There may still be remaining uncertainty around whether $h$ is safe or pivotal, but this is the only uncertainty. 
        \begin{enumerate}[label=(\alph*)]
            \item If $S(h-1)$ is Single, then $h$ is safe. Therefore, the moment $h$ is created, we can label $S(h)$ (because we can compute $P_h\big\rvert_s(h-1)$ and $Q_h\big\rvert_s(h-1)$, and flip the necessary coins).
            \item If $S(h-1)$ is Pair, and honest has published a block of height $h-1$ by this moment (possibly at exactly this moment because honest and attacker found blocks simultaneously during the round when attacker found its first block at height $h$),\footnote{Recall that honest does \emph{not} yet have a block of height $h$, as this case deals only with the case that attacker finds the first block of height $h$ alone --- we have already handled Forced Pairs in an earlier case.} then $h$ is pivotal. Therefore, the moment $h$ is created, we can label $S(h)$ as Single. 
            \item If $S(h-1)$ is Pair, and honest has \emph{not} published a block of height $h-1$, then we do not yet know whether $h$ is safe or pivotal. We learn this the moment either of the following events happens (whichever comes first):
            \begin{enumerate}
                \item Honest finds a block at height $h-1$; then $h$ is pivotal.
                \item Attacker finds a block (at height $h+1$); then $h$ is safe.
                \item Both events happen at the same time; then $h$ is safe (because there is never a moment when honest has a block of height $h-1$, but attacker does not have a block of height $h+1$).
            \end{enumerate}
        \end{enumerate}
    \end{enumerate}
    In particular, by the time there is a block of height $h+1$, we know which case we are in and thus know whether $h$ is safe or pivotal. \qedhere
\end{proof}

\begin{corollary}\label{cor:when-P-Q}
    For all $h$, we can fully determine $P_h\big\rvert_s(h-1)$ and $Q_h\big\rvert_s(h-1)$ by the time a block of height $h$ is first created.
\end{corollary}

Table~\ref{tab:P-Q} includes each of the worlds along with bounds on the values of $P_{h+1}$ and $Q_{h+1}$ conditioned on different worlds. More formally, for all $h\geq 1$, Table~\ref{tab:P-Q} lists the values of $P_{h+1}\big\rvert w$ and $Q_{h+1}\big\rvert w$ in each world $w\in W_h$. Notably, observe that these bounds do not depend on the full sequence of historical labels -- they only depend on a subset of the labels fixed by virtue of conditioning on a given world (namely a subset of $S(h-1)$ and $S(h)$). We argue in Lemma~\ref{lem:P-Q-values} that these bounds are all correct.

\begin{table}
    \caption{Some bounds on the values of $P_{h+1}$ and $Q_{h+1}$ conditioned on different worlds. Equality holds where no inequality sign is explicitly included. Bounds on $P_{h+1}$ and $Q_{h+1}$ in the third row depend on the attacker's secret lead, the distribution of which is tricky to compute. For our analysis it suffices that $P_{h+1}+Q_{h+1}\geq \alpha'+\beta'$ which we will argue holds in Lemma~\ref{lem:P-Q-values}.}\bigskip
    \label{tab:P-Q}\centering
    \begin{minipage}{.8\textwidth}
        \begin{center}
            \begin{tabular}{ ccccc }
                \toprule
                World & $S(h)$ & $P_{h+1}$ & $Q_{h+1}$ & $P_{h+1}+Q_{h+1}$ \\ \midrule
                $\mathbf{A_h}$ & Single  & $\alpha'$ & $\beta'$ & $\alpha'+\beta'$  \\
                $\mathbf{B_h}$ & Pair    & $0$ & $\beta'$ & $\beta'$ \\
                $\mathbf{C_h}$ & Single  &  & & $\geq \alpha' + \beta'$ \\
                $\mathbf{D_h}$ & Pair    & $\geq \alpha'^2$ & $\geq 0$ & $\geq \alpha'^2$\\
                $\mathbf{E_h}$ & Single  & $\alpha'$ & $\beta'$ & $\alpha'+\beta'$ \\
                $\mathbf{F_h}$ & Single  & $\alpha'$ & $\beta'$ & $\alpha'^2+\beta'$ \\
                $\mathbf{G_h}$ & Single or Pair & $\geq\alpha'+\beta'$ & $0$ & $\geq \alpha'+\beta'$\\
                $\mathbf{H_h}$ & Single or Pair & $\geq \alpha'+\beta'$ & $0$ & $\geq \alpha'+\beta'$\\
                \bottomrule
            \end{tabular}
        \end{center}
    \end{minipage}
\end{table}

\begin{lemma}\label{lem:P-Q-values}
    For all $h\geq 1$, Table~\ref{tab:P-Q} correctly lists the values of $P_{h+1}\big\rvert w$ and $Q_{h+1}\big\rvert w$ in each world in $w\in W_h$.
\end{lemma}
\begin{proof}
In world $\mathbf{A_h}$, by Definition~\ref{def:worlds} only honest creates a block at time $t(h)$. This falls into the first bullet of the generalized labeling strategy (Definition~\ref{def:labelling-general}), so $S(h)$ is labelled Single. Height $h+1$ will be a safe block iff only the attacker finds the next block (which happens with probability $\alpha'$), so $P_{h+1}=\alpha'$. Height $h+1$ will be a Forced Pair iff both honest and attacker find a block in the next round, (which happens with probability $\beta'$), so $Q_{h+1}=\beta'$.

In world $\mathbf{B_h}$, by definition both players simultaneously create blocks at time $t(h)$. Since the attacker always extends the longest chain it knows of, and this is the first moment a block of height $h$ was discovered, just prior to $t(h)$ the longest chain must have had height $h-1$, and so the two new blocks are both at height $h$. Therefore height $h$ is a Forced Pair and $S(h)$ is Pair. Height $h+1$ cannot possibly be safe, since it follows a Pair and Honest has already found a block of height $h$ before the attacker can find a block of height $h+2$, so $P_{h+1}=0$. Finally, there is a forced pair at height $h+1$ iff both players find blocks in the first round, implying $Q_{h+1}=\beta'$.

In world $\mathbf{C_h}$, by definition only the attacker creates a block at time $t(h)$, and $S(h-1)$ and $S(h)$ are both Single. We define $x$ as the probability that the attacker publishes all their hidden blocks before a block of height $h+1$ is found (that is, the probability that, if the attacker has $N$ hidden blocks labelled pair, then only Honest finds a block in each of the next $N$ rounds). Then, height $h+1$ is safe iff either (a) Attacker mines a block in any of the next $N$ rounds (with probability $1-x$) or (b) Honest mines a block in the next $N$ rounds, and only Attacker mines a block in the following round (with probability $x\cdot \alpha'$. Overall, we have $P_{h+1}=x\cdot \alpha' + 1-x$. Finally, there will be a Forced Pair at height $h+1$ iff Honest mines a block in the next $N$ rounds, and both attacker and Honest mine a block in the following round, so $Q_{h+1} = x\cdot\beta'$. The exact value of $x$ depends on the attacker's lead in hidden blocks, the distribution of which can be tricky to compute. However, this is not important for our analysis, since we only need to show that suffices for us to show $P_{h+1}+Q_{h+1}\geq \alpha'+\beta'$. We have
\[P_{h+1}+ Q_{h+1}  - (\alpha'+\beta') = x\cdot \alpha' + 1-x + x\cdot\beta'  - (\alpha'+\beta') = (1-x)(1-(\alpha'+\beta')),\]
which is in $[0,1]$ for any $x\in[0,1]$ (since $0<\alpha'+\beta'<1$).

In world $\mathbf{D_h}$, by definition only the attacker creates a block at time $t(h)$, and $S(h-1)$ is Single, and $S(h)$ is Pair. One (though not the only) way for height $h+1$ to be safe is if only the attacker finds a block in the next two rounds (which happens with probability $\alpha'$). So $P_{h+1}$ is at least $\alpha'^2$. Finally, $Q_{h+1}\geq 0$ trivially.

In world $\mathbf{E_h}$, by definition only attacker creates a block at time $t(h)$, $S(h-1)$ is Pair, and honest has published a block of height $h-1$ by time $t(h)$. Height $h$ cannot possibly be safe (since it does not follow a Single, and Honest has already found a block of height $h-1$). Thus $S(h)$ is pivotal, and Single, and the attacker immediately publishes it at $t(h)$. Height $h+1$ follows a Single and thus will be safe iff only attacker finds a block in the next round, meaning $P_{h+1} = \alpha'$. Also, height $h+1$ will be a Forced Pair iff both players find blocks in the next round, meaning $Q_{h+1}=\beta'$.

The reasoning for world $\mathbf{F_h}$ is identical to world $\mathbf{E_h}$.

In world $\mathbf{G_h}$, by definition only attacker creates a block at time $t(h)$, $S(h-1)$ is Pair, honest has \emph{not} yet published a block of height $h-1$ by time $t(h)$, and Attacker finds a block at height $h+1$ before Honest finds a block of height $h-1$. Then height $h$ is safe, and the attacker may label it Single or Pair depending on the result of its coin flips. Then height $h+1$ will be safe as long as either (a) only attacker mines in the round after $t(h+1)$ (which happens probability $\alpha'$) or (b) both Honest and Attacker mine in the round after $t(h+1)$ (which happens probability $\beta'$. Overall, $P_{h+1}\geq \alpha'+\beta'$. Finally, $Q_{h+1}=0$ since only attacker finds a block at time $t(h+1)$.

The reasoning for world $\mathbf{H_h}$ is identical to world $\mathbf{G_h}$.
\end{proof}

Validity and undetectability of this labeling strategy are shown in Lemmas~\ref{lem:validity-general}~and~\ref{lem:undetectability-general}. Lemma~\ref{lem:validity-general} will ensure that the coin flip probabilities in the third bullet are valid, and its proof will use Lemma~\ref{lem:undetectability-general} in the inductive step. Lemma~\ref{lem:undetectability-general} shows that, as long as the probabilities for all lower heights are well-defined, the label of the next height follows the correct target distribution (namely, Pair with probability exactly $\beta$ and Single otherwise). Finally, Implementability is shown in Lemma~\ref{lem:broadcast-general}.

\begin{definition}[Validity at $h$]
    We say $P$ and $Q$ are \emph{valid at $h$} if, for all state sequences $s\in\{\text{Single},\text{Pair}\}^{h-1}$, we have $\frac{\beta-Q_h\big\rvert_s}{P_h\big\rvert_s} \in[0,1]$.
\end{definition}

\begin{definition}[Undetectability at $h$]
    We say a labeling strategy is \emph{undetectable at $h$} if, for all state sequences $s\in\{\text{Single},\text{Pair}\}^{h-1}$, it produces the correct distribution of labels at height $h$; that is, if for all state sequences $s\in\{\text{Single},\text{Pair}\}^{h-1}$, we have $\Pr[S(h)=\text{Pair}\big\rvert_s]=\beta$.\footnote{We sometimes lazily write $\Pr[\cdot\big\rvert_s]$ instead of $\Pr[\cdot \mid S(i)=s_i,\ 1\leq i\leq |s|]$}
\end{definition}

\begin{lemma}[Generalized Undetectability]\label{lem:undetectability-general}
For any $h\geq 1$, if $P$ and $Q$ are valid at $h$, then the labeling strategy of Definition~\ref{def:labelling-general} is undetectable at $h$.
\end{lemma}
\begin{proof}
    Observe that by Claim~\ref{cor:when-P-Q}, the moment $S(1),\ldots S(h-1)$ are all fixed, there is enough information to determine $P_h$ and $Q_h$, which are defined conditioned \emph{only} on $S(1),\ldots S(h-1)$ and no additional information. Let us reiterate that, importantly, we are not conditioning on whether any of the previous Pairs were Forced or Chosen. This is crucial: For example, if an observer can tell there was a Forced pair at height $h-1$, they can immediately conclude the probability of having a Pair at height $h$ is exactly $\beta'<\beta$. Pulling off undetectability thus relies on the attacker's ability to properly mask whether each pair is Forced or Chosen while still achieving the correct fraction of Pairs.

    Again focus on the moment $S(1),\ldots S(h-1)$ are all fixed, namely the moment $P_h\big\rvert_s(h-1)$ and $Q_h\big\rvert_s(h-1)$ are determined. There are two possible ways for $S(h)$ to be Pair:
    \begin{enumerate}
        \item $S(h)$ is a Forced Pair (with probability $Q_h\big\rvert_s(h-1)$)
        \item $S(h)$ is a Chosen Pair because both of the following events happen 
        \begin{enumerate}
            \item the coin flip in the labeling strategy lands on Pair if Single (with probability $\frac{\beta-Q_h\big\rvert_s(h-1)}{P_h\big\rvert_s(h-1)}$)
            \item the attacker finds a safe block at height $h$ (with probability $P_h\big\rvert_s(h-1)$)
        \end{enumerate}
    \end{enumerate}
    Combining these cases, the total probability that $S(h)$ is Pair conditioned on $S(0),\ldots S(h-1)$ is exactly $\beta$ as desired.
\end{proof}

\begin{lemma}[Generalized Validity]\label{lem:validity-general}
$P$ and $Q$ are valid at $h$ for all $h\geq1$ as long as $\beta\in[\beta', \alpha'^2/2]$.
\end{lemma}

\begin{proof}
    We will proceed by induction. 

    \paragraph{Base case.} Consider first the base case where $h=1$. By definition of validity at 1, we must show $\frac{\beta-Q_1\big\rvert\text{Single}}{P_1\big\rvert\text{Single}} \in[0,1]$.
    \begin{itemize}
        \item We claim $Q_1\big\rvert \text{Single}=\beta'$. This is because the only way to get a forced pair at height 1 is if both players mine in the first round after Genesis, which happens with probability $\beta'$.
        \item We claim $P_1\big\rvert\text{Single}=\alpha'$. This is because the Genesis block is Single and thus height 1 is for sure safe, so $P_1(\text{Single})$ is exactly the probability that only the attacker mines in the first round after Genesis, which happens with probability $\alpha'$. 
    \end{itemize}
    Therefore, $\frac{\beta-Q_1\big\rvert\text{Single}}{P_1\big\rvert\text{Single}} = (\beta-\beta')/\alpha'$. Indeed,
    $(\beta-\beta')/\alpha'$ is finite since $\alpha'\not=0$, non-negative since $\beta'\leq \beta$ by assumption, and at most 1 since $(\beta-\beta')/\alpha' \leq \beta/\alpha' \leq \alpha'/2$.

    \paragraph{Inductive Step.}
    Suppose $P$ and $Q$ are valid at $h$.
    We hope to show they are valid at $h+1$. That is, we need to show that for all sequences of states $s\in\{\text{Single},\text{Pair}\}^h$, we have  $\frac{\beta-Q_{h+1}\big\rvert_s}{P_{h+1}\big\rvert_s} \in[0,1]$. Non-negativity follows from Table~\ref{tab:P-Q} and noting that in every world $Q_{h+1}\big\rvert_s\leq\beta'$, which is at most $\beta$ by assumption. It remains to show $P_{h+1}\big\rvert_s + Q_{h+1}\big\rvert_s \geq \beta$.
    
    We break down the analysis into two parts depending on the last entry in $s$. Lemmas~\ref{lem:single}~and~\ref{lem:pair} show the inequality holds when the last entry of $s$ is Single and Pair respectively. First, observe that for any $h$, each possible realization of the evolution of the blockchain corresponds to exactly one of the worlds with respect to $h$. So the following is a well-defined random variable with range $W_h$.
        
    \begin{definition}[Conditional Worlds] \label{def:conditional-worlds}
    For all heights $h$ and state sequences $s\in\{\text{Single},\text{Pair}\}^*$, let $w_h\big\rvert_s$ denote the world with respect to $h$ conditioned on the state sequence $s$. Additionally, let $x_h(s,w)=\Pr[w_h\big\rvert_s=w]$ for each world $w\in W_h$.
    \end{definition}
    
    The following observation will be useful in the proof of both.
    \begin{observation}\label{obs:convex-combination}
        For all heights $h,\geq 1$ and any state sequence $s\in\{\text{Single},\text{Pair}\}^*$, by the law of total probability we have
        \begin{align*}
            P_h(s) &= \sum_{w\in W_{h}} x_{h}(s,w) \cdot  P_h\big\rvert_{w,s} \\
            Q_h(s) &= \sum_{w\in W_{h}} x_{h}(s,w) \cdot  Q_h\big\rvert_{w,s}.
        \end{align*}
    \end{observation}

\begin{lemma}[The case where $S(h)=\text{Single}$]\label{lem:single}
    For any height $h\geq 1$ and state sequence $s\in\{\text{Single},\text{Pair}\}^{h-1}\times \{\text{Single}\}$, we have $P_{h+1}\big\rvert_s+Q_{h+1}\big\rvert_s \geq \beta.$ 
\end{lemma}
\begin{innerproof}
Applying Observation~\ref{obs:convex-combination} to height $h+1$ and a state sequence $s$ ending in $\text{Single}$, we can write
\begin{align*} 
    P_{h+1}\big\rvert_s +Q_{h+1}\big\rvert_s &= \sum_{w\in W_{h+1}} x_{h+1}(s,w) \left(P_{h+1}\big|_{w,s} + Q_{h+1}\big\rvert_{w,s} \right)\\
    &\geq  \sum_{w\in W_{h+1}} x_{h+1}(s,w) \left(\alpha'^2 + \beta' \right)\\
    &= \big(\alpha'^2 + \beta' \big) \sum_{w\in W_{h+1}} x_{h+1}(s,w) \\
    & = \alpha'^2 +\beta' \geq 2\beta+\beta' \geq \beta,
\end{align*}
where the second line follows by replacing every $\big(P_{h+1}\big\rvert_{w,s}+Q_{h+1}\big\rvert_{w,s}\big)$ with the smallest such term, which is $\alpha'^2 + \beta'$ as can be verified by referencing Table~\ref{tab:P-Q} (the columns with $S(h)=\text{Single}$ in particular).
\end{innerproof}

\begin{lemma}[The case where $S(h)=\text{Pair}$]\label{lem:pairgeneral}
    For any height $h\geq 1$ and state sequence $s\in\{\text{Single},\text{Pair}\}^{h-1}\times \{\text{Pair}\}$, we have $P_{h+1}\big\rvert_s+Q_{h+1}\big\rvert_s \geq \beta.$
\end{lemma}
\begin{innerproof}

Applying Observation~\ref{obs:convex-combination} to height $h+1$ and a state sequence $s$ ending in $\text{Pair}$, we can write
\begin{align*} 
    P_{h+1}\big\rvert_s +Q_{h+1}\big\rvert_s &= \sum_{w\in W_{h+1}} x_{h+1}(s,w) \left(P_{h+1}\big|_{w,s} + Q_{h+1}\big\rvert_{w,s} \right)\\
    &=  x_{h+1}(s,\mathbf{B_{h+1}}) \left(P_{h+1}\big|_{\mathbf{B_h},s} + Q_{h+1}\big\rvert_{\mathbf{B_h},s} \right) + \sum_{w\in W_{h+1}\setminus\{\mathbf{B_h}\} }x_{h+1}(s,w) \left(P_{h+1}\big|_{w,s} + Q_{h+1}\big\rvert_{w,s} \right)\\
    &\geq  \frac{\beta'}{\beta}\cdot \beta' + \big(1- \frac{\beta'}{\beta}\big) \alpha'^2\\
    &= \frac{\beta'^2}{\beta} + \big(1- \frac{\beta'}{\beta}\big) \alpha'^2,
\end{align*}
where the inequality uses Claim~\ref{claim:pair} to bound the first term, and the fact that as per Table~\ref{tab:P-Q}, in all state sequences ending in Pair \emph{other than} $\mathbf{B_h}$, $P_{h+1} + Q_{h+1}\geq \alpha'^2$ to bound the second term.

To wrap up, it remains to show $\frac{\beta'^2}{\beta} + \big(1- \frac{\beta'}{\beta}\big) \alpha'^2 \geq \beta$. Rearranging this inequality gives, this is equivalent to showing $\beta^2-\alpha'^2\beta \leq \beta'^2-\alpha'^2\beta'$. Since by assumption $\beta \geq \beta'$, it suffices to show the function $f(x)=x^2+\alpha'^2 x$ is decreasing. Taking derivatives with respect to $x$, the function is decreasing whenever $2x-\alpha'\leq 0$, or equivalently whenever $\beta,\beta'\leq \frac{\alpha'^2}{2}$, which is true by assumption.

\begin{claim}\label{claim:pair}
For any $h\geq 1$ and $s\in\{\text{Single},\text{Pair}\}^{h-1}\times \text{Pair}$, we have $x_{h+1}(s,\mathbf{B_{h+1}})\geq \beta'/\beta.$
\end{claim}
\begin{innerproof}
    First, we observe that the total probability that height $h$ is labeled pair is exactly $\beta$, even conditioned on $s$. To see this, recall that by the inductive hypothesis, $P$ and $Q$ are valid at $h$. Therefore, by Lemma~\ref{lem:undetectability-general}, we know that the strategy is undetectable at $h$, meaning $\Pr\big[S(h)=\text{Pair}\mid s\big]=\beta$. 

    Next, observe that $\Pr[w_h(s')=\mathbf{B_h}]\geq \beta'$ for any $s'$. This is because landing in world $\mathbf{B_h}$ requires both players to mine at height $h$ simultaneously. A prerequisite for this is that the honest player knows of the existence of all blocks at heights $h'<h$ \emph{before} anyone finds a block of height $h$. So at some point, both players had the same view of the longest chain (with no hidden pairs) and there were no blocks at height $h+1$ yet. From this state, there is a $\beta'$ probability that both players mine in the next round, which happens with probability $\beta'$.

    Applying the definition of $x_h$, the definition of conditional probability, and the previous two observations in order, we have
    \begin{align*}
        x_h(s,\mathbf{B_h}) = \Pr\big[w_h=\mathbf{B_h}\mid s \big] = \frac{ \Pr\big[w_h=\mathbf{B_h} \big\rvert_{s}]}{\Pr[S(h)=\text{Pair}]} \geq \frac{\beta'}{\beta}. \quad \qedhere
    \end{align*}
\end{innerproof}
\end{innerproof}
\end{proof}

\begin{lemma}[General Implementability]\label{lem:broadcast-general}
    The following \emph{broadcasting strategy} realizes the labels output by the labeling strategy in Definition~\ref{def:labelling-general}:
        \begin{itemize}
            \item If $S(h)$ is labeled as Single (and you created a block of height $h$): broadcast the block of height $h$ the moment a block of height $h-1$ is broadcast (and you know that $S(h)$ is labeled Single).
            \item If $S(h)$ is labeled as Pair: broadcast the block of height $h$ the moment a block of height $h$ is broadcast (and you know that $S(h)$ is labeled Pair). 
        \end{itemize}
        Note that Forced Pairs are a special case of the second bullet with no delay between the creation of the attacker block of height $h$ and its broadcasting.
\end{lemma}
\begin{proof}
    The proof is almost identical to the proof of implementability in the SP model (Lemma~\ref{lem:broadcast-SP}).
    Again, we first confirm that if the labeling strategy can be implemented, it will cause the state of every block to match its label (because there will clearly not be a chance for Honest to conflict with your Singles, and they clearly will conflict with your Pairs).
    
    The only step is to confirm that you indeed know how $h$ is labeled early enough to use this strategy. If $h$ is not a Forced Pair, the same argument in the SP model applies. If $h$ is a Forced Pair, then it is labeled as Pair the moment it is created, and its label matches its state, and  are created. 
\end{proof}

\subsection{Reward Analysis}
Recall that we're interested only in determining whether a strategy exists that is both profitable and undetectable --- we don't aim to compute the expected reward. As such, we can simplify several calculations. First, we have to nail down the expected reward achieved by being honest.

\begin{lemma} In the SP-General model where pairs naturally occur at rate $\beta'$, the attacker has an $\alpha$ fraction of the hashrate, and all honest nodes tiebreak against the attacker, the attacker wins a $x^*(\alpha,\beta'):=\frac{\alpha - \beta' + \frac{\alpha'}{1-\alpha'-\beta'}\cdot \beta'}{1-\beta'}$ fraction of blocks in the longest chain by being honest.
\end{lemma}
\begin{proof}
Observe that when the attacker is honest, the only Pairs that occur are natural. Every Pair will be followed by some number of other Pairs (possibly zero), and then a Single. Whoever wins this Single will determine who wins the original Pair. Conditioned on being a Single, it will be created by honest with probability $\frac{1-\alpha'-\beta'}{1-\beta'}$ and attacker with probability $\frac{\alpha'}{1-\beta'}$. Therefore, each Pair is won by the attacker with probability $\frac{\alpha'}{1-\alpha'-\beta'}$. Lemma~\ref{lem:pair} then suffices to conclude the lemma statement.
\end{proof}

\begin{proposition}\label{prop:maingeneral} In the SP-General model where pairs naturally occur at rate $\beta'$, the attacker has an $\alpha$ fraction of the hashrate, and all honest nodes tiebreak against the attacker, a sufficient condition for our specific Selfish Mining strategy with target Pair-rate $\beta$ to be strictly profitable compared to being honest is:
\begin{align*}
&\Bigg(\frac{(2-\beta')\cdot (1-\alpha'-\beta')\cdot \alpha'}{(1-\beta')^2}+2\beta' + \frac{(2+\beta'-\alpha')\cdot \alpha'}{1-\alpha'}-1\Bigg)\\
&-\Bigg(\frac{\alpha' - (\alpha')^2}{1-\alpha'-\beta'}\Bigg)\cdot \Bigg(\frac{(2-\beta')\cdot (1-\alpha'-\beta')}{1-\beta'} + 2\beta' + \frac{(2+\beta'-\alpha')\cdot \alpha'}{(1-\alpha')}-1 \Bigg) \geq 0
\end{align*}
\end{proposition}

In particular, observe that substituting in $\beta'=0$ (implying $\alpha'=\alpha$) recovers Theorem~\ref{thm:main1}. Moreover, when $\alpha \geq 0.382$,\footnote{Note that $\alpha \geq 0.382$ does not immediately determine $\alpha'$, as $\alpha'$ is a function of both $\alpha$ and $\beta'$.} the inequality holds for all $\beta' \in [0,1]$.

Note also that the bound in Proposition~\ref{prop:maingeneral} does not depend on $\beta$. This is intended, as our analysis proceeds by lower bounding our reward (for any $\beta$) by that of an alternate strategy that is independent of $\beta$.

\begin{proof}[Proof of Proposition~\ref{prop:maingeneral}]

In Lemma~\ref{lem:sapirshtein}, via direct application of a result in \cite{SapirshteinSZ16}, we argue that the attacker's strategy outperforms honest if and only if it gets positive expected reward in a decision process that gets $+1$ for every attacker block in the longest chain and $-x^*(\alpha,\beta')$ for every block in the longest chain. We will not include the proof and refer the reader to \cite{SapirshteinSZ16} instead. Intuitively, Lemma~\ref{lem:sapirshtein} follows by observing that when the attacker has produced a $\delta$ fraction of $n$ blocks in the longest chain, the attacker loses $n\lambda$ and gains $n\delta$. This is positive if and only if $\delta > \lambda$.

\begin{lemma}[rephrased from {\cite[Corollary 10]{SapirshteinSZ16}}]\label{lem:sapirshtein}
 Consider rewarding the Attacker $+1$ for every block they find that becomes permanently in the longest chain, and $-\lambda$ for every block that is permanently in the longest chain.\footnote{This includes blocks produced by both Attacker and Honest.} A strategy wins a $\geq \lambda$ fraction of blocks in the longest chain if and only if it receives non-negative expected reward by this measure.
\end{lemma}

In particular, this implies that our proposed strategy is strictly profitable if and only if it has positive expected reward when $\lambda = x^*(\alpha,\beta')$.

The next corollary clarifies the situation in which Attacker's strategy deviates from honest mining.

\begin{corollary}
     The labeling strategy of Definition~\ref{def:labelling-general} outperforms honest if and only if its expected reward in $G$ exceeds $1-x^*(\alpha,\beta')$ \emph{after hiding its first block until it has no more hidden blocks}.
\end{corollary}
\begin{proof}
    Consider the case when there are no hidden blocks. When Honest finds a block, both the honest strategy and our strategy lose $x^*(\alpha,\beta')$. When Attacker finds a block, the honest strategy gains $1-x^*(\alpha,\beta')$, and we immediately return to a state where there are no hidden blocks. Our strategy sometimes does this as well (in which case it remains tied with the honest strategy). Or, it sometimes hides this block and continues being strategic until it eventually returns to a state with no hidden blocks. Conditioned on making this decision, it must gain reward at least $1-x^*(\alpha,\beta')$ in expectation to outperform the honest strategy (because it gets identical expected reward to the honest strategy in all other cases, and therefore must outperform it in this case in order to strictly profit).
\end{proof}

The next observation allows us to focus on the reward gained by a simpler strategy than ours. Specifically, we consider the following strategy.

\begin{definition}[Short Selfish Mining] The \emph{Short Selfish Mining} strategy does the following, after hiding its first block:
\begin{itemize}
\item If Honest finds the next block, followed by $k$ Natural Pairs (perhaps $k=0$), followed by Single-Attacker, immediately publish all hidden blocks.
\item If Honest finds the next block, followed by $k$ Natural Pairs (perhaps $k=0)$, followed by Single-Honest, accept the loss of the longest chain and return to the state with no hidden blocks.
\item If the next round is Natural Pair, immediately publish both blocks.
\item If the next round is Single-Attacker, followed by $k$ Single-Attackers (perhaps $k=0$), followed by a Natural Pair or Single-Honest \emph{immediately publish all blocks}. 
\end{itemize}
Observe that the difference to our strategy is the final bullet -- rather than make use of our lead to possibly get more blocks in the longest chain, we take a smaller victory immediately.
\end{definition}

Indeed, our strategy is identical in all but the final bullet. In the case of the final bullet, our strategy gets \emph{at least} as many blocks in the longest chain as Short Selfish Mining, and perhaps more, and therefore gets strictly more reward.\footnote{This claim holds when considering the fraction of the longest chain won by our strategy versus the simpler strategy, but this is not trivial to see. It is easy to see that our strategy achieves greater expected reward by the counting of~\cite{SapirshteinSZ16} -- our strategy simply adds more blocks to the longest chain before resetting.}

\begin{observation} conditioned on hiding its first block, our strategy strictly outperforms Short Selfish Mining.
\end{observation}

The plan from here is to analyze the expected reward gained by the short selfish mining strategy. In particular, we will show that it gets reward $> 1-x^*(\alpha,\beta')$.

Consider the case when the short selfish mining strategy has just hidden a block. Let's count blocks in the following manner:
\begin{itemize}
\item If the next round is Single-Honest, followed by $k$ Natural Pairs (perhaps $k=0$), followed by Single-Honest, then there are $k+2$ blocks entering the longest chain, none of which belong to the attacker. This contributes $-(k+2)x^*(\alpha,\beta')$, and occurs with probability $(1-\alpha'-\beta')^2\cdot (\beta')^k$.
\item If the next round is Single-Honest, followed by $k$ Natural pairs (perhaps $k=0$), followed by Single-Attacker, then there are $k+2$ blocks entering the longest chain, all of which belong to the attacker. This contributes $(k+2)-(k+2)x^*(\alpha,\beta')$, and occurs with probability $(1-\alpha'-\beta')\cdot (\beta')^k \cdot \alpha'$.
\item If the next round is Natural Pair, then there are $2$ blocks entering the longest chain, all of which belong to the attacker. This contributes $2-2x^*(\alpha,\beta')$, and occurs with probability $\beta'$.
\item If the next round is Single-Attacker, followed by $k$ Single-Attackers (perhaps $k=0$), followed by a Natural Pair, then there are $2+k+1$ blocks entering the longest chain, all of which belong to the attacker. This contributes $(k+3) - (k+3)x^*(\alpha,\beta')$, and occurs with probability $(\alpha')^{k+1}\cdot \beta'$.
\item If the next round is Single-Attacker, followed by $k$ Single-Attackers (perhaps $k=0$), followed by a Single-Honest, then there are $k+2$ blocks entering the longest chain, all of which belong to the attacker. This contributes $(k+2)-(k+2)x^*(\alpha,\beta')$, and occurs with probability $(\alpha')^{k+1}\cdot (1-\alpha'-\beta')$.
\end{itemize}

So, summing the contribution of all cases together, we get:

\begin{align*}
-&\sum_{k=0}^\infty (k+2)\cdot x^*(\alpha,\beta')\cdot(1-\alpha'-\beta')^2\cdot (\beta')^k\\
+&\sum_{k=0}^\infty \left(k+2 - (k+2)\cdot x^*(\alpha,\beta')\right)\cdot(1-\alpha'-\beta')\cdot (\beta')^k\cdot \alpha'\\
+&2\beta' - 2\cdot x^*(\alpha,\beta')\cdot \beta'\\
+&\sum_{k=0}^\infty \left(k+3 - (k+3)x^*(\alpha,\beta')\right)\cdot (\alpha')^{k+1}\cdot \beta'\\
+&\sum_{k=0}^\infty \left(k+2 - (k+2)x^*(\alpha,\beta')\right)\cdot (\alpha')^{k+1}\cdot (1-\alpha'-\beta')\\
=-& x^*(\alpha,\beta')\cdot(1-\alpha'-\beta')^2\cdot \frac{2-\beta'}{(1-\beta')^2}\\
+& \frac{2-\beta'}{(1-\beta')^2}\cdot (1-\alpha'-\beta')\cdot \alpha' - x^*(\alpha,\beta')\cdot \frac{2-\beta'}{(1-\beta')^2}\cdot (1-\alpha'-\beta')\cdot \alpha'\\
+&2\beta' - 2\cdot x^*(\alpha,\beta')\cdot \beta'\\
+&\frac{3-2\alpha'}{(1-\alpha')^2}\cdot \alpha' \cdot \beta' - x^*(\alpha,\beta') \cdot \frac{3-2\alpha'}{(1-\alpha')^2}\cdot \alpha' \cdot \beta'\\
+&\frac{2-\alpha'}{(1-\alpha')^2}\cdot \alpha' \cdot (1-\alpha'-\beta') - x^*(\alpha,\beta') \cdot \frac{2-\alpha'}{(1-\alpha')^2}\cdot \alpha' \cdot (1-\alpha'-\beta')\\
=~& \left( \frac{(2-\beta') \cdot (1-\alpha' - \beta')\cdot \alpha'}{(1-\beta')^2}+2\beta' + \frac{(3-2\alpha')\cdot \alpha' \cdot \beta' + (2-\alpha')\cdot \alpha' \cdot (1-\alpha'-\beta')}{(1-\alpha')^2}\right)\\
-&x^*(\alpha,\beta')\cdot \Bigg(\frac{(1-\alpha'-\beta')^2 \cdot (2-\beta') +(2-\beta')\cdot (1-\alpha'-\beta')\cdot \alpha')}{(1-\beta')^2} +2\beta' \\
&\qquad \qquad \qquad+\frac{(3-2\alpha')\cdot \alpha' \cdot \beta' + (2-\alpha')\cdot \alpha' \cdot (1-\alpha'-\beta')}{(1-\alpha')^2}\Bigg)\\
=~& \Bigg(\frac{(2-\beta')\cdot (1-\alpha'-\beta')\cdot \alpha'}{(1-\beta')^2}+2\beta' + \frac{(2+\beta'-\alpha')\cdot \alpha'}{1-\alpha'}\Bigg)\\
&-x^*(\alpha,\beta')\cdot \Bigg(\frac{(2-\beta')\cdot (1-\alpha'-\beta')}{1-\beta'} + 2\beta' + \frac{(2+\beta'-\alpha')\cdot \alpha'}{(1-\alpha')} \Bigg)
\end{align*}

Recall that our goal is to understand conditions on $\alpha, \beta'$ so that the above quantity exceeds $1-x^*(\alpha,\beta')$. That is, we know that the Short Selfish Mining strategy (with natural pairs at rate $\beta'$ and attacker hashrate $\alpha$) is strictly profitable compared to honest if and only if:

\begin{align*}
&\Bigg(\frac{(2-\beta')\cdot (1-\alpha'-\beta')\cdot \alpha'}{(1-\beta')^2}+2\beta' + \frac{(2+\beta'-\alpha')\cdot \alpha'}{1-\alpha'}\Bigg)\\
&-x^*(\alpha,\beta')\cdot \Bigg(\frac{(2-\beta')\cdot (1-\alpha'-\beta')}{1-\beta'} + 2\beta' + \frac{(2+\beta'-\alpha')\cdot \alpha'}{(1-\alpha')} \Bigg) \geq 1-x^*(\alpha,\beta')\\
\Leftrightarrow &\Bigg(\frac{(2-\beta')\cdot (1-\alpha'-\beta')\cdot \alpha'}{(1-\beta')^2}+2\beta' + \frac{(2+\beta'-\alpha')\cdot \alpha'}{1-\alpha'}-1\Bigg)\\
&-x^*(\alpha,\beta')\cdot \Bigg(\frac{(2-\beta')\cdot (1-\alpha'-\beta')}{1-\beta'} + 2\beta' + \frac{(2+\beta'-\alpha')\cdot \alpha'}{(1-\alpha')}-1 \Bigg) \geq 0.
\end{align*}

So, the final step from here is to rearrange the above to find a cleaner condition on $\alpha,\beta'$.

\begin{align*}
&\Bigg(\frac{(2-\beta')\cdot (1-\alpha'-\beta')\cdot \alpha'}{(1-\beta')^2}+2\beta' + \frac{(2+\beta'-\alpha')\cdot \alpha'}{1-\alpha'}-1\Bigg)\\
&-x^*(\alpha,\beta')\cdot \Bigg(\frac{(2-\beta')\cdot (1-\alpha'-\beta')}{1-\beta'} + 2\beta' + \frac{(2+\beta'-\alpha')\cdot \alpha'}{(1-\alpha')}-1 \Bigg) \geq 0\\
\Leftrightarrow &\Bigg(\frac{(2-\beta')\cdot (1-\alpha'-\beta')\cdot \alpha'}{(1-\beta')^2}+2\beta' + \frac{(2+\beta'-\alpha')\cdot \alpha'}{1-\alpha'}-1\Bigg)\\
&-\Bigg(\alpha + \alpha \beta' - \beta' 
+ \frac{\alpha' \cdot \beta'}{1-\alpha'-\beta'}\Bigg)\cdot \Bigg(\frac{(2-\beta')\cdot (1-\alpha'-\beta')}{1-\beta'} + 2\beta' + \frac{(2+\beta'-\alpha')\cdot \alpha'}{(1-\alpha')}-1 \Bigg) \geq 0\\
\Leftrightarrow&\Bigg(\frac{(2-\beta')\cdot (1-\alpha'-\beta')\cdot \alpha'}{(1-\beta')^2}+2\beta' + \frac{(2+\beta'-\alpha')\cdot \alpha'}{1-\alpha'}-1\Bigg)\\
&-\Bigg(\frac{\alpha'+\beta'}{1+\beta'}\cdot (1+\beta')-\beta'
+ \frac{\alpha' \cdot \beta'}{1-\alpha'-\beta'}\Bigg)\cdot \Bigg(\frac{(2-\beta')\cdot (1-\alpha'-\beta')}{1-\beta'} + 2\beta' + \frac{(2+\beta'-\alpha')\cdot \alpha'}{(1-\alpha')}-1 \Bigg) \geq 0\\
\Leftrightarrow&\Bigg(\frac{(2-\beta')\cdot (1-\alpha'-\beta')\cdot \alpha'}{(1-\beta')^2}+2\beta' + \frac{(2+\beta'-\alpha')\cdot \alpha'}{1-\alpha'}-1\Bigg)\\
&-\Bigg(\alpha'
+ \frac{\alpha' \cdot \beta'}{1-\alpha'-\beta'}\Bigg)\cdot \Bigg(\frac{(2-\beta')\cdot (1-\alpha'-\beta')}{1-\beta'} + 2\beta' + \frac{(2+\beta'-\alpha')\cdot \alpha'}{(1-\alpha')}-1 \Bigg) \geq 0\\
\Leftrightarrow&\Bigg(\frac{(2-\beta')\cdot (1-\alpha'-\beta')\cdot \alpha'}{(1-\beta')^2}+2\beta' + \frac{(2+\beta'-\alpha')\cdot \alpha'}{1-\alpha'}-1\Bigg)\\
&-\Bigg(\frac{\alpha' - (\alpha')^2}{1-\alpha'-\beta'}\Bigg)\cdot \Bigg(\frac{(2-\beta')\cdot (1-\alpha'-\beta')}{1-\beta'} + 2\beta' + \frac{(2+\beta'-\alpha')\cdot \alpha'}{(1-\alpha')}-1 \Bigg) \geq 0.
\end{align*}

\end{proof}
\newpage
\section{Omitted Proofs}\label{app:omitted}

\begin{proof}[Proof of Proposition~\ref{prop:SPOK}]
We will show how to couple an \nmgl\ with hashrates $\vec{\alpha}$ with one  with hashrates $\langle\alpha_1,\frac{1-\prod_{i=2}^n (1-\alpha_i \cdot \ell))}{\ell} \rangle$. Observe that in the first case, the probability that only Miner $1$ produces a block is $\alpha':=\frac{\alpha_1\cdot \ell \cdot \prod_{i=2}^n (1-\alpha_i\cdot \ell)}{1-\prod_{i=1}^n (1-\alpha_i\cdot \ell)}$, the probability that both Miner $1$ and a Miner $>1$ produces a block is $\beta':=\frac{\alpha_1\cdot \ell \cdot (1-\prod_{i=2}^n (1-\alpha_i\cdot \ell))}{1-\prod_{i=1}^n (1-\alpha_i\cdot \ell)}$, and the probability that only Miners $>1$ produce a block is $\frac{(1-a_1\cdot \ell)\cdot (1-\prod_{i=2}^n(1-a_i \cdot \ell))}{1-\prod_{i=1}^n (1-\alpha_i\cdot \ell)}=\frac{1-\prod_{i=1}^n (1-a_i\cdot \ell) }{1-a_i\cdot \ell-\prod_{i=1}^n (1-\alpha_i\cdot \ell)} = 1-\frac{a_i\cdot \ell}{1-\prod_{i=1}^n (1-a_i\cdot \ell) } = 1-\alpha'-\beta'$.

Consider now \nmgl\ with hashrates $\langle \alpha_1,\frac{1-\prod_{i=2}^n (1-\alpha_i \cdot \ell))}{\ell}\rangle$. Then again we have that the probability that only Miner $1$ produces a block is $\alpha':=\frac{\alpha_1\cdot \ell \cdot \prod_{i=2}^n (1-\alpha_i\cdot \ell)}{1-\prod_{i=1}^n (1-\alpha_i\cdot \ell)}$, the probability that both Miner $1$ and a Miner $>1$ produces a block is $\beta':=\frac{\alpha_1\cdot \ell \cdot (1-\prod_{i=2}^n (1-\alpha_i\cdot \ell))}{1-\prod_{i=1}^n (1-\alpha_i\cdot \ell)}$, and the probability that only Miners $>1$ produce a block is $\frac{(1-a_1\cdot \ell)\cdot (1-\prod_{i=2}^n(1-a_i \cdot \ell))}{1-\prod_{i=1}^n (1-\alpha_i\cdot \ell)}=\frac{1-\prod_{i=1}^n (1-a_i\cdot \ell) }{1-a_i\cdot \ell-\prod_{i=1}^n (1-\alpha_i\cdot \ell)} = 1-\frac{a_i\cdot \ell}{1-\prod_{i=1}^n (1-a_i\cdot \ell) } = 1-\alpha'-\beta'$.

Therefore, the games can be coupled so that in every round whether just Miner $1$, Miner $1$ and other miners, or just other miners produce a block is identical. Because $s$ is SP-Simple, it will take the same action across both coupled games. Because $s'$ is a longest-chain protocol that tiebreaks lexicographically or reverse lexicographically, it will also take the same action in both games (specifically, it will tiebreak for or against $1$ the same way in both games). This immediately establishes that the reward of Miner $1$ is the same in both games.

To see that statistical undetectability translates between the two games, observe that statistical undetectability exactly states that \nmgl\ with strategies $s, \vec{s}_{-i}$ can be coupled with \nmglp\ with a longest-chain strategy and $\vec{s}_{-i}$ and hashrates $\vec{\alpha}$. By the work above, this latter game can be coupled with \nmglp\ with a longest-chain strategy and $s'$ and hashrates $\langle \alpha_1,\frac{1-\prod_{i=2}^n (1-\alpha_i \cdot \ell))}{\ell}\rangle$ (because all longest-chain strategies are SP-Simple). Therefore, if either of the two undetectability claims hold, all four games can be coupled as desired (implying that the other undetectability claim holds as well).
\end{proof}

\end{document}